\documentclass[runningheads,orivec,envcountsame,envcountsect]{llncs}

\usepackage{enumitem}
\setlist[itemize]{topsep=2pt, parsep=0pt, partopsep=0pt, itemsep=2pt}
\setlist[enumerate]{topsep=2pt, parsep=0pt, partopsep=0pt, itemsep=2pt}

\usepackage[T1]{fontenc}
\usepackage[utf8]{inputenc}

\spnewtheorem*{restatetheorem}{Theorem}{\bfseries}{}
\spnewtheorem*{restatelemma}{Lemma}{\bfseries}{}
\spnewtheorem*{restatecorollary}{Corollary}{\bfseries}{}

\usepackage[b]{esvect} 
\usepackage{graphicx}
\graphicspath{{figures/}}
\usepackage{amssymb}
\usepackage{amsmath}
\allowdisplaybreaks
\usepackage{proof}
\usepackage{tikz}
\usetikzlibrary{cd,positioning,decorations.text,decorations.pathmorphing}
\tikzcdset{scale cd/.style={every label/.append style={scale=#1},cells={nodes={scale=#1}}}}
\usepackage{tikz-cd}
\usepackage[only=llparenthesis,rrparenthesis,llbracket,rrbracket]{stmaryrd}
\usepackage{qcircuit}

\usepackage{color}  
\definecolor{darkblue}{rgb}{0,0,0.5}
\definecolor{darkgreen}{rgb}{0,0.4,0}
\usepackage[
  bookmarks,
  colorlinks=true,     
  linkcolor=darkblue,
  citecolor=darkgreen,
  urlcolor=blue,
  bookmarksopen,
  bookmarksopenlevel=2,
  bookmarksdepth=2,
  bookmarksnumbered=true
]{hyperref}

\usepackage[capitalise,nameinlink]{cleveref}
\Crefname{section}{Section}{Sections}
\Crefname{figure}{Figure}{Figures}
\Crefname{table}{Table}{Tables}
\Crefname{theorem}{Theorem}{Theorems}
\Crefname{lemma}{Lemma}{Lemmas}
\Crefname{definition}{Definition}{Definitions}

\makeatletter
\providecommand{\qed}{\hbox{\rule{1ex}{1ex}}}
\newcommand{\qedhere}{%
  \ifmmode
    \tag*{\qed}
  \else
    \hfill\qed
  \fi
}
\makeatother

\newcommand\ket[1]{\ensuremath{|#1\rangle}}

\newcommand\Span[1]{\ensuremath{{\mathsf{Span}}{#1}}}

\newcommand\comp[2][]{#2^{\bot^{#1}}}
\newcommand\Rpart[1]{\mathsf{Re(#1)}}
\newcommand\Ipart[1]{\mathsf{Im(#1)}}

\newcommand\Var{\ensuremath{\mathsf{Var}}}
\newcommand\Val{{\s V}}
\DeclareRobustCommand{\ValD}{\ensuremath{\vv{\mathsf{V}}}}

\newcommand\lra{\longrightarrow}

\newcommand\ansubst[2]{\ensuremath{\langle #1 \rangle_{#2}}}

\newcommand\AbsBasis{\ensuremath{\mathbb{A}}}
\newcommand\dom[1]{\mathrm{dom}(#1)}
\newcommand\sdom[1]{\mathrm{dom}^{\sharp}(#1)}
\newcommand\FV[1]{\mathrm{FV}(#1)}

\def\R{\mathbb{R}}            
\def\C{\mathbb{C}}            
\def\Val{\mathrm{V}}          
\def\Sph{\mathcal{S}_1}       
\def\scal#1#2{\langle{#1}~|~{#2}\rangle}

\def\<{\langle}
\def\>{\rangle}
\def\Pair#1#2{(#1,#2)} 
\def\Lam#1#2#3{\lambda#1^{#2}{.}#3} 
\def\letkeyword{\mathsf{let}}
\def\inkeyword{\mathsf{in}}

\def\LetP#1#2#3#4#5#6{\letkeyword\,\Pair{#1^{#2}}{#3^{#4}}=#5~\inkeyword~#6}

\def\case#1#2#3#4#5{\ensuremath{\mathsf{case}~#1~\mathsf{of}~\{#2\mapsto #4 \mid #3\mapsto #5\}}}
\def\gencase#1#2#3#4#5{\ensuremath{\mathsf{case}~#1~\mathsf{of}~\{#2\mapsto #4 \mid \dotsb \mid #3\mapsto #5\}}}

\def\Kron#1#2{\ensuremath{\delta_{#1,#2}}}

\def\lraneq{\rightsquigarrow}
\def\eval{\lra^*}
\def\lave{\mathrel{\reflectbox{$\eval$}}}

\def\Arr{\Rightarrow}
\def\Type{\mathbb{T}}
\def\BasisType{\Type_\flat}
\def\sem#1{\llbracket#1\rrbracket}

\def\SUB#1#2{#1\le#2}

\def\TYP#1#2#3{#1~{\vdash}~#2~{:}~#3}
\def\SORTH#1#2#3#4{#1~{\vdash}~#2\perp#3~{:}~#4}
\def\ORTH#1#2#3#4#5#6{#1~{\vdash}~(#2~{\vdash}~#3)\perp(#4~{\vdash}~#5)~{:}~#6}

\def\snam#1{\textsc{\scriptsize\upshape(#1)}}
\def\real{\Vdash}

\def\ds{\displaystyle}
\outer\long\def\COUIC#1{}
\def\sqrthalf{{\textstyle\frac{1}{\sqrt{2}}}}

\newcommand\B{\mathbb B}
\newcommand\XB{\mathbb X}
\newcommand\Hd{\mathsf{H}}

\newcommand{\cnot}[2]{\mathsf{CNOT}\ #1\ #2}

\newcommand{\pauliX}[1]{\mathsf{NOT}\ #1}
\newcommand{\pauliZXB}{\mathsf{Z}_{\XB}}
\newcommand{\cnotXB}[2]{\mathsf{CNOT}_{\XB}\ #1\ #2}

\newcommand{\pauliXXB}[1]{\mathsf{NOT}_{\XB}\ #1}

\newcommand{\Bell}{\mathsf{Bell}}
\newcommand{\lambdaB}{\lambda_B}
\newcommand\basis[1]{\ensuremath{\flat_{#1}}}
\newcommand\genbasis[3]{\ensuremath{\basis{\{#1\}_{#2}^{#3}}}}

\begin{document}

\title{Basis-Sensitive Quantum Typing via Realisability}

\author{
  Alejandro Díaz-Caro\inst{1,2} 
  \and
  Octavio Malherbe\inst{3} 
  \and
  Rafael Romero\inst{2,4} 
}

\authorrunning{A. Díaz-Caro, O. Malherbe, and R. Romero}

\institute{
  Université de Lorraine, CNRS, Inria, LORIA, France
  \and
  Universidad Nacional de Quilmes, Argentina
  \and
  Universidad de la República, Facultad de Ingeniería, IMERL, Uruguay
  \and
  ICC, CONICET-Universidad de Buenos Aires, Argentina\\
}

\maketitle 

\begin{abstract}
  We present $\lambdaB$, a quantum‐control $\lambda$‐calculus that refines
  previous basis-sensitive systems by allowing abstractions to be expressed with
  respect to arbitrary---possibly entangled---bases.  Each abstraction and
  $\mathsf{let}$ construct is annotated with a basis, and a new basis-dependent
  substitution governs the decomposition of value distributions.  These
  extensions preserve the expressive power of earlier calculi while enabling
  finer reasoning about programs under basis changes.  A realisability semantics
  connects the reduction system with the type system, yielding a direct
  characterisation of unitary operators and ensuring safety by construction.
  From this semantics we derive a validated family of typing rules, forming the
  foundation of a type-safe quantum programming language.  We illustrate the
  expressive benefits of $\lambdaB$ through examples such as Deutsch’s algorithm
  and quantum teleportation, where basis-aware typing captures classical
  determinism and deferred-measurement behaviour within a uniform framework.

  \keywords{Quantum computing \and Realisability semantics\and $\lambda$-calculus.}
\end{abstract}

\section{Introduction}
The no-cloning theorem \cite{WoottersZurek1982} and the no-deleting theorem
\cite{PatiBraunstein2000} are two well-known results in quantum mechanics that
state it is impossible to copy or delete an arbitrary qubit. 
There is, however, a subtlety: although arbitrary qubits
cannot be copied or deleted, this is possible for known---and, in the case of
deletion, separable---qubits. This implies that qubits with known values behave
as classical data and can be treated accordingly. Moreover, it suffices to know
the basis to which a qubit belongs in order to copy and, to some extent, delete it.

In most quantum programming languages, qubits are defined with respect to a
canonical basis—often referred to as the computational basis. In this setting,
classical bits correspond to the basis vectors, whereas qubits are unit-norm
linear combinations of them. Classical bits can be copied and deleted freely,
while such operations on arbitrary qubits are restricted.

In this paper we introduce a quantum $\lambda$-calculus within the
quantum-control paradigm—by contrast with the classical-control one, where the
control flow is classical and cannot be superposed. Our approach is inspired by
a line of work on basis-sensitive quantum typing. The earliest system,
Lambda-S \cite{DiazcaroDowekRinaldiBIO19}, could distinguish whether a qubit
was in the computational basis, allowing duplication and erasure only in that
case. Later, Lambda-S$_1$
\cite{DiazcaroGuillermoMiquelValironLICS19,DiazCaroMalherbe2022} refined this
approach by restricting to unit-norm vectors, introducing higher-order
abstractions, and ensuring that terms of type qubit-to-qubit denote unitary
maps. More recently, the Lambda-SX calculus \cite{DiazcaroMonzonAPLAS25}
generalised Lambda-S to arbitrary non-entangled single-qubit bases, albeit
through a purely syntactic framework restricted to first order. The calculus we
present here extends these ideas to a unit-norm, higher-order setting over
multiple-qubit bases, grounded in a realisability semantics.

The realisability methodology, originating with Kleene's work on Heyting
arithmetic \cite{KleeneJSL45}, provides a constructive framework that connects
operational semantics with type systems. In our context, it allows the
extraction of a sound type system directly from the reduction semantics of the
calculus, ensuring that safety properties hold by construction. The approach
proceeds as follows:
\begin{enumerate}
  \item Define a calculus with a deterministic evaluation strategy.
  \item Define types as sets of closed values in the language.
  \item Define the typing judgement so that asserting that a term has a given
  type is, by definition, to state that the term reduces to a value of that
  type.
\end{enumerate}
Each typing rule therefore corresponds to a provable theorem in this setting.
Rather than building ad hoc typing rules, the system is derived from the
computational content of the calculus, making it possible to define whole
families of type systems by proving the validity of new rules.

Following this approach, we enrich abstractions with explicit basis
decorations. Intuitively, the reduction system treats values expressed in the
chosen basis as classical data, while linear combinations of these values
represent quantum data and reduce linearly within the term. This refinement
enables duplication and erasure for qubits in known bases while maintaining
linear handling for unknown qubits.

The objective of this work is twofold. First, to employ the extracted type
system to provide a more precise description of programs. Second, to take
advantage of the extended syntax to express quantum algorithms in a flexible
and compositional way, rather than merely translating circuits.

The idea of tracking basis information within quantum programming languages has
appeared in several forms across the literature.  
The work in~\cite{Perdrix2008} proposed an abstract model for static
entanglement analysis, where the basis of qubits is used to control duplication
and track non-entanglement properties.  
The approach in~\cite{DiazcaroMonzonAPLAS25} introduced a typed
$\lambda$-calculus that incorporates basis annotations into the type system,
ensuring linear handling of quantum data while providing strong
meta-theoretical guarantees.  
Our system generalises these ideas by allowing abstractions to range over
arbitrary---possibly entangled---bases, extending basis awareness beyond the
single-qubit setting and recovering higher-order computation.

A number of other frameworks support reasoning about multiple bases through
different paradigms.  
The ZX-calculus~\cite{CoeckeRossICALP08} provides a graphical language where
spiders correspond to computations relative to the computational and diagonal
bases, while the Many-Worlds
Calculus~\cite{ChardonnetdeVismeValironVilmartLMCS25} also accommodates
superpositions of programs but within a diagrammatic semantics.  
By contrast, our approach operates directly within a typed term calculus,
allowing basis transitions to be tracked syntactically.

From a categorical standpoint, the framework
in~\cite{HeunenKaarsgaard2021} models measurement and decoherence through
\emph{quantum information effects}, capturing basis change as a semantic
effect.  
Similarly,~\cite{CaretteJeunenKaarsgaardSabry2024} introduces
\textsc{Quantum}$\Pi$, a language combining two interpretations of a reversible
classical calculus—one per basis—via an effect construction.  
In both cases, basis sensitivity arises semantically.  
In contrast, $\lambdaB$ integrates it directly at the syntactic and typing
levels, ensuring that basis transformations are explicit in program structure.

Other languages aim to unify classical and quantum computation.
For instance,~\cite{VoichickLiRandHicks2023} generalises classical
control constructs and interprets duplication and discarding semantically as
entanglement and partial trace.  
Our calculus follows the opposite philosophy: linearity and duplicability are
syntactically restricted by the basis-sensitive type system, guaranteeing
coherent quantum control without relying on external semantic constraints.

Finally, related efforts in semantic characterisation—such as the fully
abstract model of~\cite{ClairambaultdeVisme2019} or the dual calculi
described in~\cite{ChoudhuryGay2025}—approach the logical foundations of
quantum computation from complementary angles.  
In contrast, $\lambdaB$ focuses on the syntactic distinction of bases and
superpositions within a unified quantum $\lambda$-calculus, providing a concrete
basis-sensitive typing and higher-order quantum control.

The contributions of this paper can be summarised as follows.
\begin{itemize}
  \item We introduce a \emph{basis-sensitive} quantum $\lambda$-calculus in the
  quantum-control paradigm, extending previous systems to a setting that is
  unit-norm, higher-order, and supports multiple-qubit entangled bases.
\item We develop a \emph{realisability-based semantics} that connects
  reduction and typing, enabling the systematic extraction of a sound
  type system.
  \item We formalise a \emph{type algebra} capable of tracking basis
  information throughout programs, and derive typing rules corresponding to
  provable theorems in the realisability interpretation.
  \item We illustrate the expressive power of the calculus through examples
  involving basis-dependent operations and multi-qubit structures, showing how
  it subsumes and extends previous single-basis approaches.
\end{itemize}

The paper is organised as follows.  
\Cref{sec:calculus} introduces the core language, its syntax, congruence rules,
and basis-dependent substitution.  
\Cref{sec:reduction} defines the operational semantics, and
\Cref{sec:model} presents the realisability model with unitary type semantics,
the characterisation of unitary operators, and the typing rules.  
\Cref{sec:examples} presents some representative examples.
All omitted proofs appear in the appendices, and
\Cref{sec:conclusion} concludes with final remarks and perspectives.

\section{Core language}\label{sec:calculus}
This section presents the calculus on which our realisability model will be
built. It is a $\lambda$-calculus extended with linear combinations of terms, in the line of~\cite{ArrighiDowekLMCS17}.

The syntax of the calculus is described in
\Cref{tab:Syntax}. The calculus is divided into four
distinct syntactic categories: \emph{pure values} ($\Val$), \emph{pure terms}
($\Lambda$), \emph{value distributions} ($\ValD$), and \emph{term
distributions} ($\vec\Lambda$). Values are composed of variables, decorated
lambda abstractions, and two basis values representing orthogonal vectors,
$\ket 0$ and $\ket 1$. A pair of values is also considered a value.  Terms
include values, applications, pair constructors and destructors, and pattern
matching for orthogonal vectors, represented by the $\mathsf{case}$ operator.
Both term and value distributions are built as $\C$-linear combinations of
terms and values, respectively. 
We use $v,u,w$ to denote values and $t,s,r$ for terms, writing $\vv{\cdot}$ when
they are distributions.
We also consider the following notation for linear combinations of
pairs. We stress that this notation for pairs does not appear in the syntax,
but is rather useful to describe specific states.
\[
  \Pair{\alpha  t+\vv{t_1}}{\vv{t_2}} 
  := \alpha\Pair{t}{\vv{t_2}} + \Pair{\vv{t_1}}{\vv{t_2}}
  \qquad\qquad
  \Pair{r}{\alpha  t+\vv{t_1}} 
  := \alpha\Pair{r}{t} + \Pair{r}{\vv{t_1}}
\]

The subsets $B$, $B_1$, and $B_2$ appearing in the abstractions and pair
destructors denote the bases of the vector spaces in which the terms are
expressed; their precise nature is made explicit below.
Before formally defining the notion of bases, we first establish the
algebraic structure underlying the space of value distributions.
This is achieved by introducing the congruence relation defined in
\Cref{tab:Congruence}.
The congruence captures the intended behaviour of addition and scalar
multiplication, allowing us to treat linear combinations of terms and values
as genuine algebraic entities rather than mere syntactic constructions.
In particular, it enforces the commutativity and associativity of addition and
the distributivity of scalar multiplication, thereby justifying summation
notation~$\sum$.

\begin{table}[t]
  \begin{align*}
    v &::= x \mid \Lam{x}B{\vv{t}} \mid (v, v) \mid \ket{0} \mid \ket{1} &
    (\Val)\\
    t &::= w \mid tt \mid (t,t) \mid
    \LetP{x}{B_1}{y}{B_2}{\vv{t}}{\vv{t}}\mid
    \gencase{\vv{t}}{\vv{v}}{\vv{v}}{\vv{t}}{\vv{t}} &
    (\Lambda) \\
    \vv{v} &::= v \mid \vv{v}+\vv{v} \mid \alpha \vv{v} \mid \vv 0
    \qquad\hfill(\alpha\in\C) & (\ValD) \\
    \vv{t} &::= t \mid \vv{t}+\vv{t} \mid \alpha \vv{t} \mid \vv 0
    \qquad\hfill(\alpha\in\C) & (\vv \Lambda)
  \end{align*}
  \caption{Syntax of the calculus, where $B, B_1, B_2 \subseteq \ValD$.}
  \label{tab:Syntax}
\end{table}
\begin{table}[t]
  \begin{align*}
    0\vv t &\equiv \vv 0 & \vv t+\vv 0 &\equiv \vv t \\
    1\vv{t} &\equiv \vv{t} &
    \alpha(\beta\vv{t}) &\equiv (\alpha\beta)\vv{t} \\
    \vv{t_1}+\vv{t_2} &\equiv \vv{t_2}+\vv{t_1} &
    (\vv{t_1}+\vv{t_2})+\vv{t_3} &\equiv \vv{t_1}+(\vv{t_2}+\vv{t_3}) \\
    (\alpha+\beta)\vv{t} &\equiv \alpha\vv{t}+\beta\vv{t} &
    \alpha(\vv{t_1}+\vv{t_2}) &\equiv \alpha\vv{t_1}+\alpha\vv{t_2} \\
    \vv{t}(\alpha\vv{s}) &\equiv \alpha(\vv{t}\vv{s}) &
    (\alpha\vv{t})\vv{s} &\equiv \alpha(\vv{t}\vv{s}) \\
    (\vv{t}+\vv{s})\vv{r} &\equiv \vv{t}\vv{r}+\vv{s}\vv{r} &
    \vv{t}(\vv{s}+\vv{r}) &\equiv \vv{t}\vv{s}+\vv{t}\vv{r} &&
  \end{align*}
  \vspace{-2\baselineskip}
  \begin{align*}
    &\LetP{x_1}{X}{x_2}{Y}{(\alpha\vv{t})}{\vv{s}}
    \equiv \alpha(\LetP{x_1}{X}{x_2}{Y}{\vv{t}}{\vv{s}}) \\
    &\LetP{x_1}{X}{x_2}{Y}{\vv{t}+\vv{s}}{\vv{r}}
    \equiv
      (\LetP{x_1}{X}{x_2}{Y}{\vv{t}}{\vv{r}})
      +(\LetP{x_1}{X}{x_2}{Y}{\vv{s}}{\vv{r}}) 
      \\
    &\gencase{\alpha \vv{t}}{\vv{v_1}}{\vv{v_n}}{\vv{s_1}}{\vv{s_n}}
    \equiv \alpha(\gencase{\vv{t}}{\vv{v_1}}{\vv{v_n}}{\vv{s_1}}{\vv{s_n}}) \\
    &\gencase{(\vv{t}+\vv{s})}{\vv{v}}{\vv{w}}{\vv{r_1}}{\vv{r_2}}
    \equiv\begin{aligned}[t]
      &\gencase{\vv{t}}{\vv{v}}{\vv{w}}{\vv{r_1}}{\vv{r_2}}\\
      &+\gencase{\vv{s}}{\vv{v}}{\vv{w}}{\vv{r_1}}{\vv{r_2}}
    \end{aligned}
  \end{align*}
  \caption{Term congruence}
  \label{tab:Congruence}
\end{table}

The rules in \Cref{tab:Congruence} ensure that the
set of value distributions satisfies the axioms of a vector space.  The notion
of basis in this calculus builds on this algebraic foundation. Once the vector
structure is established, we can define which subsets $B\subseteq \ValD$
qualify as \emph{bases}, thereby justifying the decorations appearing in the
syntax of \Cref{tab:Syntax}.

Before proceeding further, let us briefly illustrate the intuition behind the
congruence in \Cref{tab:Congruence}.  
The key idea is that arguments are decomposed over the bases associated with
their corresponding abstractions.  
As in linear algebra, a vector can be rewritten as a linear combination of
basis elements; for instance,
\(
  (1,0)
  = \tfrac{1}{\sqrt{2}}\!\left(\tfrac{(1,1)}{\sqrt{2}}
  + \tfrac{(1,-1)}{\sqrt{2}}\right)
\)
expresses $(1,0)$ in the basis
$\{\tfrac{(1,1)}{\sqrt{2}}, \tfrac{(1,-1)}{\sqrt{2}}\}$.
By identifying $(1,0)$ with $\ket{0}$ and $(0,1)$ with $\ket{1}$, we see that
the calculus allows every vector (or value distribution) to be expressed as a
superposition of elements of the basis attached to each abstraction.
The congruence ensures that these linear combinations behave algebraically as
in a complex vector space, supporting the standard operations of addition,
scalar multiplication, and the zero vector.

To properly characterise the sets that decorate the lambda abstractions, we
must first specify the kind of values they contain.
\begin{definition}[Qubits]\label{def:Qubit}
  A \emph{one-dimensional qubit} is a value distribution of the form
  $\alpha\ket{0} + \beta\ket{1}$, where $|\alpha|^2 + |\beta|^2 = 1$. An
  \emph{$n$-dimensional qubit} is a value distribution of the form
  $\alpha\Pair{\ket{0}}{\vv{w_1}} + \beta\Pair{\ket{1}}{\vv{w_2}}$, where
  $\vv{w_1}$ and $\vv{w_2}$ are $(n-1)$-dimensional qubits, and $\alpha$ and
  $\beta$ satisfy the same normalisation condition.
\end{definition}
We use the usual Dirac shorthand $\ket{xy}$ for $(\ket{x},\ket{y})$, and
extend it to longer tuples by associating to the right.

From now on, we call the value distributions, i.e. the elements of $\ValD$,
\emph{vectors}. The vector space $\ValD$ is equipped with an inner
product defined by,  
\(
  \scal{\vv{v}}{\vv{w}} := \sum_{i=1}^n\sum_{j=1}^m
  \overline{\alpha_i}\,\beta_j\,\delta_{v_i,w_j},
\)
and an $\ell_2$-norm 
\(
  \|\vv{v}\| := \sqrt{\scal{\vv{v}}{\vv{v}}}
  = \sqrt{\sum_{i=1}^n|\alpha_i|^2}
\),
where $\vv{v}=\sum_{i=1}^n\alpha_i v_i$ and
$\vv{w}=\sum_{j=1}^m\beta_j w_j$, and $\delta_{v_i,w_j}$ is the Kronecker
delta, equal to $1$ if $v_i=w_j$ and $0$ otherwise.

With this notion of inner product, we can complete the description of the
calculus syntax. As expected, two values are \emph{orthogonal} when their inner
product is equal to zero. We can now formally define the sets that
decorate abstractions.
\begin{definition}[Basis]\label{def:NthDimensionalBasis}
  A set of value distributions $B$ is an \emph{$n$-dimensional orthonormal
  basis} if it satisfies:
  \begin{enumerate}
    \item Each element of $B$ is an $n$-dimensional qubit \emph{(cf.~\Cref{def:Qubit})}.
    \item Distinct elements of $B$ are pairwise orthogonal.
  \end{enumerate}
\end{definition}

From now on, the syntax introduced in \Cref{tab:Syntax}
is restricted to sets $B$ forming $n$-dimensional
orthonormal bases.

Unlike standard orthonormal bases, we require that their elements be qubits
(rather than variables or abstractions).  
These sets indicate the basis in which a term is expressed: qubits in the
decorating basis are treated call-by-value, while others are decomposed as
$\C$-linear combinations of basis elements, with the function applied linearly
to each component.  If a term cannot be expressed in the decorating basis,
evaluation becomes stuck.

As in classical linear algebra, no non-trivial linear combination of basis
elements yields the null vector; otherwise, some basis vector would violate
pairwise orthogonality. Consequently, each decomposition over a basis is
unique.

\begin{theorem}[Unique decomposition]\label{thm:UniqueDecomposition}
  If $B$ is an $n$-dimensional basis, then every $n$-dimensional qubit has a
  unique decomposition over $B$.
  \qed
\end{theorem}

\begin{corollary}[Preservation under congruence]\label{cor:EquivalentDecomposition}
  If $\vv{v} \equiv \vv{w}$, then they share the same decomposition over any
  basis $B$.
  \qed
\end{corollary}

We denote the computational basis by $\B=\{\ket{0},\ket{1}\}$ and the diagonal
basis by $\XB=\{\ket{+},\ket{-}\}$, where
$\ket{\pm}=\tfrac{1}{\sqrt{2}}(\ket{0}\pm\ket{1})$.  

As usual, the expression $\vv t[\vv v/x]$ denotes the
usual capture-avoiding substitution of $\vv v$ for $x$ in $\vv t$.
However, beta-reduction depends on the basis chosen for the abstraction, so we must
define a substitution that takes this mechanism into account. Intuitively, this
operation substitutes variables with vectors expressed in the chosen basis; the
accompanying coefficients are those of the value distribution being
substituted.

Alongside this substitution, we introduce a special basis, denoted
$\AbsBasis$, which acts as the canonical basis for $\lambda$-abstractions. In
this way, we restrict function distributions to a single admissible basis.

\begin{definition}[Basis-dependent substitution]
  Let $\vv t$ be a term distribution, $\vv v$ a value distribution, $x$ a
  variable, and $B$ an orthonormal basis. We define the substitution
  $\vv t\ansubst{\vv v/x}{B}$ as follows:
  \[
    \vv t\ansubst{\vv v/x}{B} \;=\;
    \begin{cases}
      \ds\sum_{i\in I}\alpha_i\,\vv t\,[\vv{b_i}/x] &
        \text{if } B=\{\vv{b_i}\}_{i\in I}\ \text{and}\
        \vv v \equiv \ds\sum_{i\in I}\alpha_i\,\vv{b_i},\\[6pt]
      \ds\sum_{i\in I}\alpha_i\,\vv t\,[v_i/x] &
        \text{if } B=\AbsBasis\ \text{and}\
        \vv v = \ds\sum_{i\in I}\alpha_i\,v_i,\\[2pt]
      \text{undefined} & \text{otherwise.}
    \end{cases}
  \]
\end{definition}

The two cases in the definition capture distinct substitution modes.  
In the first, substitution proceeds linearly using the decomposition of
$\vv v$ over the explicit basis $B$.  
In the second, when $B=\AbsBasis$, substitution proceeds linearly over the pure
values that form $\vv v$.  
This case recovers the substitution mechanism of 
non base-sensitive calculi---as first defined in
\cite{ArrighiDowekLMCS17}---but generalises it by introducing
basis-dependent behaviour.  
This special case is also the only one applicable to
$\lambda$-abstractions, since abstractions do not belong to orthonormal bases.

  The definition also extends to pairs of values. If
  $\vv v=\sum_{i\in I}\alpha_i\,\Pair{\vv{v_i}}{\vv{w_i}}$,
  with $\vv{v_i}\in\Span(B_1)$ and $\vv{w_i}\in\Span(B_2)$,
  then
  \[
    \vv t\ansubst{\vv v/x\otimes y}{B_1\otimes B_2}
      \;=\; \sum_{i\in I}\alpha_i\,
      \bigl(\vv t\ansubst{\vv{v_i}/x}{B_1}\ansubst{\vv{w_i}/y}{B_2}\bigr),
  \]
  where $B_1$ and $B_2$ are (orthonormal) bases---or $\AbsBasis$---associated
  with $x$ and $y$, respectively; the symbol $\otimes$ in $B_1\otimes B_2$ is
  purely notational.

\begin{example}
  Let
  $\vv v = \alpha\ket{01}
           + \beta\ket{10}$.
  Then 
  \[
    (y,x)\ansubst{\vv v/x\otimes y}{\B\otimes\B}
    = \alpha\,(y,x)[\ket{0}/x][\ket{1}/y]
    + \beta\,(y,x)[\ket{1}/x][\ket{0}/y] 
    = 
    \alpha\ket{10}
    + \beta\ket{01}
  \]
\end{example}

With this substitution in place, we can establish certain needed properties.
First, basis-dependent substitution distributes over linear combinations.

\begin{lemma}[Distributivity over linear combinations]\label{lem:distributiveSubstitution}
  For term distributions $\vv{t_i}$, a value distribution $\vv{v}$, a
  variable $x$, coefficients $\alpha_i\in\C$, and a basis $B$ such that
  $\ansubst{\vv v/x}{B}$ is defined:
  \(
    \Bigl(\sum_i \alpha_i\vv{t_i}\Bigr)\ansubst{\vv v/x}{B}
    \equiv
    \sum_i \alpha_i\vv{t_i} \ansubst{\vv v/x}{B}
  \).
  \qed
\end{lemma}

The next result states that substitution behaves consistently within each
equivalence class induced by the congruence $\equiv$.

\begin{lemma}[Compatibility with congruence]\label{lem:EquivSubstitutions}
  For value distributions $\vv{v},\vv{w}$, a term distribution $\vv{t}$, and
  an orthonormal basis $B$ such that both
  $\ansubst{\vv{v}/x}{B}$ and $\ansubst{\vv{w}/x}{B}$ are defined:
  if $\vv{v}\equiv\vv{w}$, then
  $\vv{t}\ansubst{\vv{v}/x}{B}
  =\vv{t}\ansubst{\vv{w}/x}{B}$.
  \qed
\end{lemma}

Remark that the property in \Cref{lem:EquivSubstitutions} does not extend across
  different bases; that is,
  $\vv{t}\ansubst{\vv{v}/x}{A}\not\equiv\vv{t}\ansubst{\vv{v}/x}{B}$.
  For example,
  \[
    (\Lam{x}{C}{y})\ansubst{\ket{+}/y}{\XB}
    = \Lam{x}{C}{\ket{+}} 
    \not\equiv
    \tfrac{1}{\sqrt{2}}\big((\Lam{x}{C}{\ket{0}})
    + (\Lam{x}{C}{\ket{1}})\big)
    = (\Lam{x}{C}{y})\ansubst{\ket{+}/y}{\B}.
  \]
  This difference arises because the relation $\equiv$ does not commute with
  lambda abstraction, nor with the case construct. Although the two terms are
  operationally equivalent, the calculus distinguishes between the
  superposition of results,
  $\Lam{x}{B}{\alpha\vv{v_1} + \beta\vv{v_2}}$,
  and the superposition of functions,
  $\alpha(\Lam{x}{B}{\vv{v_1}}) + \beta(\Lam{x}{B}{\vv{v_2}})$.
  This distinction reflects a physical intuition: the former corresponds to a
  single experiment producing a superposition of outcomes, while the latter
  represents a superposition of distinct experiments.

Finally, we introduce a convenient notation for generalised substitutions over
a term by closed values. A substitution $\sigma$ can be seen as a finite set of
individual substitutions applied consecutively to a term. Formally, for a term
$\vv{t}$, closed value distributions $\vv{v_1},\dots,\vv{v_n}$, variables
$x_1,\dots,x_n$, and bases $B_1,\dots,B_n$:
\[
  \vv{t}\ansubst{\sigma}{}
  := \vv{t}\ansubst{\vv{v_1}/x_1}{B_1}\dotsb\ansubst{\vv{v_n}/x_n}{B_n}.
\]
Since each $\vv{v_i}$ is closed, the order of substitutions is irrelevant. We
regard $\sigma$ as a partial function from variables to pairs of closed value
distributions and bases, and write $\dom{\sigma}$ for its domain. The operation
extends naturally: for a term $\vv{t}$, substitution $\sigma$, a new variable
$x\notin\dom{\sigma}$, value distribution $\vv{v}$, and basis $B$,
\(
  \vv{t}\ansubst{\sigma}{}\ansubst{\vv{v}/x}{B}
  = \vv{t}\ansubst{\sigma'}{}
\),
where $\sigma'$ extends $\sigma$ by mapping $x$ to $(\vv v,B)$. Likewise, two
disjoint substitutions $\sigma_1$ and $\sigma_2$ can be merged:
\(
  \vv{t}\ansubst{\sigma_1}{}\ansubst{\sigma_2}{}
  = \vv{t}\ansubst{\sigma'}{}
\),
where $\dom{\sigma_1}\cap\dom{\sigma_2}=\emptyset$ and $\sigma'$ coincides with
$\sigma_i$ on $\dom{\sigma_i}$ for $i=1,2$.

\section{Operational semantics}\label{sec:reduction}

The reduction system interprets every vector relative to the basis attached to
its abstraction, allowing a step only when the argument can be decomposed on
that basis.  
\Cref{tab:Reduction} defines the elementary relation~$\lraneq$, while terms are
considered modulo the congruence of \Cref{tab:Congruence}.  
Hence the effective reduction, written~$\lra$, is defined modulo~$\equiv$:
a step $\vv t\lra\vv r$ abbreviates $\vv t\equiv\vv t'\lraneq\vv r'\equiv\vv r$.

\begin{table}[t]
  \begin{align*}
    \text{If }\vv{t}\ansubst{\vv v/x}{X}\text{ is defined,}\quad
    (\Lam{x}{X}{\vv{t}})\vv{v}
    &\lraneq \vv{t}\ansubst{\vv v/x}{X}\\
    \text{If }\vv{t}\ansubst{\vv v/x}{X\otimes Y}\text{ is defined,}\quad
    \LetP{x}{X}{y}{Y}{\vv v}{\vv{t}}
    &\lraneq \vv{t}\ansubst{\vv{v}/x\otimes y}{X\otimes Y}\\
    \gencase{\sum_{i=1}^{n}\alpha_i \vv{v_i}}{\vv{v_1}}{\vv{v_n}}{\vv{t_1}}{\vv{t_n}}
    &\lraneq \sum_{i=1}^{n}\alpha_i \vv{t_i}
  \end{align*}
  \[
    \begin{array}{c}
      \infer{st\lraneq s\vv r}{t\lraneq \vv r}
      \qquad\qquad
      \infer{tv\lraneq \vv rv}{t\lraneq\vv r}
      \qquad\qquad
      \infer{\alpha t\lraneq \alpha \vv r}{t\lraneq\vv r & \alpha\neq 0}
      \qquad\qquad
      \infer{t+\vv s\lraneq\vv r+\vv s}{t\lraneq\vv r}
      \\[5pt]
      \infer{\LetP{x}{A}{y}{B}{t}{\vv{s}}\lraneq
      \LetP{x}{A}{y}{B}{\vv r}{\vv{s}}}{t\lraneq \vv r} 
      \\[5pt]
      \infer{\gencase{\vv t}{\vv v}{\vv w}{\vv{s_1}}{\vv{s_n}}\lraneq
      \gencase{\vv r}{\vv v}{\vv w}{\vv{s_1}}{\vv{s_n}}}{t\lraneq \vv r}
    \end{array}
  \]
  \caption{Reduction system}
  \label{tab:Reduction}
\end{table}

The side condition~$\alpha\neq0$ avoids vacuous steps such as
$0t\lraneq0r$, preventing spurious nondeterminism and preserving the
determinism of~$\lraneq$.

The main reduction rules are $\beta$-reduction, $\mathsf{let}$ binding, and
$\mathsf{case}$ pattern matching.  
Both $\lambda$ and $\mathsf{let}$ bind variables decorated with an orthonormal
basis, indicating which vectors are treated as classical data.  Linear
combinations of these vectors are handled as quantum data, reducing linearly by
the congruence rules of \Cref{tab:Congruence}, so that
$t(\alpha\vv s+\beta\vv r)$ is equivalent to
$\alpha\,t\vv s+\beta\,t\vv r$.

The only exception occurs for higher-order terms: since no orthogonal bases are
defined for abstractions, we introduce a special basis~$\AbsBasis$ acting as
the canonical one for higher-order values—intuitively, the set of all pure
values.  For instance,
\[
  (\Lam{x}{\AbsBasis}{\vv t})
  \sum_{i=1}^{n}\alpha_i(\Lam{y}{X}{\vv{s_i}})
  \lra
  \sum_{i=1}^n\alpha_i\vv t[\Lam{y}{X}{\vv{s_i}}/x].
\]

The $\mathsf{case}$ pattern matching controls program flow.  
Each operator keeps track of a set of orthogonal values and tests whether the
argument equals each vector, selecting the matching branch.  
If the argument is a linear combination of several vectors, the result is the
corresponding linear combination of the branches.  For example:
\[
  \case{\ket{-}}{\ket{0}}{\ket{1}}{\vv{t_1}}{\vv{t_2}} \lra
  \tfrac{1}{\sqrt{2}}\,\vv{t_1} - \tfrac{1}{\sqrt{2}}\,\vv{t_2}.
\]

The advantage over a conditional branching is the ability to match against
several vectors simultaneously.  For boolean tuples this makes no difference,
as each component can be treated independently.  However, some orthogonal bases
cannot be expressed as products of smaller ones.  This general
$\mathsf{case}$ allows us to match directly against those vectors.  For
example, using the \emph{Bell basis}\footnote{ The four Bell states are
$\Phi^{\pm}=\sqrthalf(\ket{00}\pm\ket{11})$ and
$\Psi^{\pm}=\sqrthalf(\ket{01}\pm\ket{10})$.
}:
\[
  \mathsf{case}\;\vv{v}\;\mathsf{of}\;\{
  \Phi^+ \mapsto \vv{t_1}\ \mid
  \Phi^- \mapsto \vv{t_2}\ \mid
  \Psi^+ \mapsto \vv{t_3}\ \mid
  \Psi^- \mapsto \vv{t_4}\ \}.
\]

This Bell basis is central in quantum communication. In
\Cref{sec:teleportation}, we revisit the quantum teleportation protocol,
which relies heavily on these states.

Defining the system in this way yields a strategy within the
\emph{call-by-value} family, namely a generalisation of the
\emph{call-by-basis} strategy introduced in~\cite{ArrighiDowekLMCS17} and
further analysed in~\cite{AssafDiazcaroPerdrixTassonValironLMCS14}. Whereas
call-by-basis fixes a single computational basis for evaluation, our variant,
which we call \emph{call-by-arbitrary-basis}, allows each abstraction to attach
its own orthonormal basis to its argument. Evaluation remains weak: no
reduction takes place under $\lambda$, pairs, $\mathsf{let}$, or
$\mathsf{case}$ constructors.

The congruence relation on terms gives rise to different redexes. We write
$\eval$ for the reflexive–transitive closure of $\lra$. We can show that the
equivalence relation $\equiv$ commutes with $\eval$; in other words,
equivalence is preserved by reduction modulo~$\equiv$.

\begin{theorem}[Reduction preserves equivalence]\label{thm:confluence}
  Let $\vv{t}$ and $\vv{s}$ be closed term distributions with
  $\vv{t}\equiv\vv{s}$. If $\vv{t}\lraneq\vv{t'}$ and $\vv{s}\lraneq\vv{s'}$,
  then there exist term distributions $\vv{r_1}$ and $\vv{r_2}$ such that
  $\vv{t'}\eval\vv{r_1}$, $\vv{s'}\eval\vv{r_2}$, and
  $\vv{r_1}\equiv\vv{r_2}$.
  Diagrammatically:
  
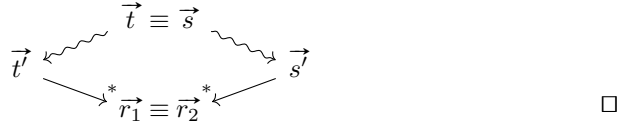
\[
    \begin{tikzcd}[row sep=1pt,baseline=(current bounding box.south)]
      & \vv{t}
        \arrow[ld,decorate,decoration={snake, amplitude=0.8, segment length=6pt}, ->]
        &[-3em] \equiv
        &[-3em] \vv{s}
        \arrow[rd,decorate,decoration={snake, amplitude=0.8, segment length=6pt}, ->]
        &\\
      \vv{t'}\arrow[dr,"*",pos=0.9] & & & &
      \vv{s'}\arrow[ld,"*"',pos=0.9] \\
      & \vv{r_1} & \equiv & \vv{r_2} & 
    \end{tikzcd}
    \tag*{\smash{\raisebox{.6\baselineskip}{\qed}}}
  \]
\end{theorem}

\begin{remark}\label{rmk:determinism}
  Since the reduction relation $\eval$ is defined modulo~$\equiv$, the result
  above is equivalent to stating that there exists a single distribution
  $\vv{r}$ such that $\vv{t'}\eval\vv{r}$ and $\vv{s'}\eval\vv{r}$.
  Moreover, as the elementary reduction~$\lraneq$ is deterministic, the
  reduction relation~$\lra$ is also deterministic; that is, if
  $\vv{t}\lra\vv{r_1}$ and $\vv{t}\lra\vv{r_2}$, then
  $\vv{r_1}\equiv\vv{r_2}$.
\end{remark}

\section{Realizability model}\label{sec:model}

\subsection{Unitary type semantics}
Given the deterministic machine presented in the previous section
(see~\Cref{rmk:determinism}), the next
step towards extracting a typing system is to define the sets of values that
characterise its types. To achieve this, we first need to identify the notion
of a type.

Our aim is to define types exclusively inhabited by values of norm~1. The
vectors we wish to study all belong to the \emph{unit sphere}. We write $\Sph$
for the set $\Sph := \{\vv v \in \vv{\Val} \mid \|\vv v\| = 1\}$, which
corresponds to the mathematical representation of quantum data as unit vectors
in a Hilbert space.

\begin{definition}[Unitary value distribution]
  A value distribution $\vv{v}$ is said to be \emph{unitary} if its norm equals~1,
  that is, if $\vv{v}\in\Sph$.
\end{definition}

\begin{definition}[Unitary type]
  A \emph{unitary type} (or simply a \emph{type}) is a notation~$A$ together
  with a set of unitary value distributions, denoted~$\sem{A}$, called the
  \emph{unitary semantics} of~$A$.
\end{definition}

We now turn to type realizers.  
Since the global phase of a quantum state has no physical significance, terms
that differ only by a phase should share the same type.  
Thus, we assign identical types to $\vv t$ and $e^{i\theta}\vv t$, a principle
that guides the definition of realizers.

\begin{definition}[Type realizer]\label{def:Realizer}
  Given a type~$A$ and a term distribution~$\vv t$, we say that~$\vv t$
  \emph{realizes}~$A$ (written~$\vv t \real A$) if there exists a value
  distribution~$\vv v\in\sem A$ such that
  $\vv{t}\eval e^{i\theta}\vv{v}$ for some $\theta\in\R$,
\end{definition}

\begin{table}[t]
  \[
    A := \basis{X} \mid A\Arr A \mid A\times A \mid \sharp A,
    \qquad\text{where $X$ is any orthonormal basis.}
  \]
  \begin{align*}
    \sem{\basis{X}}&:= X\\
    \sem{A\times B}&:= \bigl\{\,(\vv v, \vv w)\;\bigm|\; \vv v\in\sem{A},~\vv w\in\sem{B}\,\bigr\}\\
    \sem{A\Arr B}&:=
    \bigl\{\,\sum_{i=1}^{n}\alpha_i(\Lam{x}{X}{\vv{t_i}})\in\Sph
      \;\bigm|\;
      \forall\vv{w}\in\sem{A},\,
      (\sum_{i=1}^{n}\alpha_i \vv{t_i})\ansubst{\vv{w}/x}{X}\real B
    \bigr\}\\
    \sem{\sharp A}&:= {(\sem{A}^\bot)}^\bot\quad\text{Where: }A^\bot = \{\vec{v}\in\Sph\,\mid\,\scal{\vec{v}}{a} = 0, \forall a\in A\}
  \end{align*}
  \caption{Type notations and semantics}
  \label{tab:UnitaryTypes}
\end{table}

With the notions of unitary types and their realizers in place, we can now
define a concrete approach for our language. We begin with the type grammar
given in~\Cref{tab:UnitaryTypes} and build a simple algebra from the sets of
values we aim to represent.
Before that, we introduce the notion of orthogonal complement, which will be
used in the semantics of the~$\sharp$~type:
\[
  \comp{A} = \{\, \vv{v}\in \Sph \mid \scal{\vv{v}}{a} = 0 \text{ for all } a\in A\,\}.
\]

The types $\basis{X}$ serve as atomic types. Each of them represents a finite
set~$X$ of orthogonal vectors forming an orthonormal basis. For instance, a
boolean type can be represented by a basis of size~2, yet we are not restricted
to a single one, since there are infinitely many bases to choose from.

Pairs are standard, and written using the notation introduced earlier.
The arrow type~$A\Arr B$ consists of distributions of $\lambda$-abstractions
that map elements of~$\sem{A}$ to realizers of~$B$.  
Finally, the type~$\sharp A$ denotes the double orthogonal complement of~$A$
intersected with the unit sphere.
It represents quantum data---linear resources that cannot be erased or
duplicated.  
Intuitively, applying the~$\sharp$ operator to a type~$A$ yields the span of the
original interpretation (intersected with the unit sphere).  This captures the
possible linear combinations of values in the unitary semantics of~$A$, as
stated in the following theorem.

\begin{theorem}\label{thm:SharpCharacterization}
  The interpretation of a type~$\sharp A$ contains precisely the
  norm-$1$ linear combinations of values in~$\sem{A}$:
  \(
    \sem{\sharp A}
    = (\sem{A}^\bot)^\bot
    = \Span(\sem{A}) \cap \Sph
  \).
  \qed
\end{theorem}

The following theorem shows that, as expected for a span, multiple
applications of the~$\sharp$ operator have no further effect beyond the first
application.

\begin{theorem}\label{thm:IdempotentSharp}
  The~$\sharp$ operator is idempotent; that is,
  $\sem{\sharp A} = \sem{\sharp(\sharp A)}$.
  \qed
\end{theorem}

\begin{remark}
  A basis type~$\basis{X}$ may consist of value distributions of pairs and can
  therefore be written as the product type of smaller bases. For example, if
  $X=\{\ket{00},\ket{01},\ket{10},\ket{11}\}$, then $\sem{\basis{X}}=\sem{\basis{\B}\times\basis{\B}}$.
  However, this is not possible for entangled bases. A clear example is the
  Bell basis.
\end{remark}

It remains to verify that our type algebra indeed captures the intended sets of
value distributions. The following theorem shows that every member of a type
interpretation has norm~$1$.

\begin{theorem}\label{prop:UnitaryTypes}
  For every type~$A$, $\sem{A}\subseteq\Sph$.
  \qed
\end{theorem}

Defining types as sets of values naturally induces a semantic notion of
subtyping. We say that a type~$A$ is a subtype of a type~$B$
(written~$A\leq B$) when the set of realizers of~$A$ is included in that of~$B$.
If the two sets coincide, we say that $A$ and $B$ are \emph{isomorphic}
(written~$A\cong B$).

\begin{example}
  For every type~$A$, we have $A\leq\sharp A$.
  For the base types $\basis{\B}$ and $\basis{\XB}$, however,
  neither inclusion holds:
    $\basis{\B}\not\leq\basis{\XB}$ and
    $\basis{\XB}\not\leq\basis{\B}$.
  Nevertheless, their linear extensions coincide,
  since $\sharp\basis{\B}\cong\sharp\basis{\XB}$.
\end{example}

Although every type is defined by norm-1 value distributions, not every
norm-1 distribution is contained in the interpretation of a type.
For example, consider the distribution
$\tfrac{1}{\sqrt{2}}(\ket{0} + \ket{00})$.
Another example is a linear combination of abstractions defined over different
bases. For instance, the term
\[
  \tfrac{1}{\sqrt{2}}(\Lam{x}{{\B}}{\pauliX{x}})
  + \tfrac{1}{\sqrt{2}}(\Lam{x}{{\XB}}{x})
\]
is not a member of an arrow type, since the bases decorating each abstraction
do not match. However, it is computationally equivalent to the abstraction
$(\Lam{x}{{\B}}{\ket{+}})$, which in turn belongs to the set
$\sem{\basis{\B}\Arr\basis{\XB}}$.

We denote by~$\Type$ the set of all types and by~$\BasisType$ the set of all
basis types~$\basis{X}$.

\subsection{Characterisation of unitary operators}

One of the main results of~\cite{DiazcaroGuillermoMiquelValironLICS19}
is the characterisation of $\C^2\to\C^2$ unitary operators using values in
$\sem{\sharp\basis{\B}\Arr\sharp\basis{\B}}$~\cite[Theorem IV.12]{DiazcaroGuillermoMiquelValironLICS19}.
In this subsection we extend this result. Our goal is to prove that abstractions
of type $\sharp\basis{X}\Arr\sharp\basis{Y}$, where both bases have size~$n$,
represent unitary operators
$\C^n\to\C^n$.

Unitary operators are isomorphisms between Hilbert spaces, as they preserve the
structure of the space. With this in mind, the first step is to show that the
members of $\sem{\sharp\basis{X}\Arr\sharp\basis{Y}}$ map basis vectors from
$\sem{\basis{X}}$ onto orthogonal vectors in $\sem{\sharp\basis{Y}}$. In other
words, these abstractions preserve both norm and orthogonality.

\begin{lemma}\label{lem:BasesIso}
  Let $X$ and $Y$ be orthonormal bases of the same finite
  dimension, and let $\Lam{x}{{X}}{\vv t}$ be a closed $\lambda$-abstraction.
  Then $\Lam{x}{{X}}{\vv t}\in\sem{\sharp\basis{X}\Arr\sharp\basis{Y}}$
  if and only if 
  for all $\vv{v_i},\vv{v_j}\in\sem{\basis{X}}$,
  there exist value distributions
  $\vv{w_i},\vv{w_j}\in\sem{\sharp\basis{Y}}$ such that,
    $\vv{t}[\vv{v_i}/x]\eval\vv{w_i}$
    and
    $\vv{t}[\vv{v_j}/x]\eval\vv{w_j}$,
    with
    $\vv{w_i}\perp\vv{w_j}$ whenever $i\neq j$.
    \qed
\end{lemma}

Terms such as $\ket{0}$ and $\ket{1}$ are syntactic objects of the
calculus, not vectors of~$\C^2$. Nevertheless, when discussing the behaviour of
terms on Hilbert spaces, we shall occasionally abuse notation and identify
value distributions representing qubits with their corresponding canonical
basis vectors in~$\C^n$. This identification applies only to those
distributions that denote quantum data, not to general syntactic values such as
$\lambda$-abstractions---in particular, to all elements of
$\sem{\sharp\basis{X}}$ for any orthonormal basis~$X$. This identification can
be made precise as follows.

\begin{definition}
  Let $X$ be an orthonormal basis of size~$n$. For every $\vv{v}\in X$, we can
  write
  \(
    \vv{v}\equiv \sum_{i=1}^{n}\alpha_i\ket{i}
  \),
  where $\ket{i}$ denotes the $n$-qubit tuple of $\ket{0}$ and $\ket{1}$
  corresponding to the binary representation of~$i$, and
  $\sum_{i=1}^{n}|\alpha_i|^2=1$. For example, for $n=4$, $\ket{3}$ is
  $\ket{0011}$. We then define $\pi_n:X\to\C^n$ by
  \(
    \pi_n(\vv{v}) = (\alpha_1,\dotsc,\alpha_n)
  \).
\end{definition}

To lighten notation, we shall henceforth omit~$\pi_n$ and use~$\vv v$
directly.

We now establish a correspondence between $\lambda$-abstractions and operators
on~$\C^n$. Intuitively, an abstraction represents a linear operator when its
operational behaviour coincides with the action of that operator on vectors.
Formally:

\begin{definition}
  A $\lambda$-abstraction $\Lam{x}{{X}}{\vv{t}}$ is said to \emph{represent}
  an operator $F:\C^n\to\C^n$ if
    $(\Lam{x}{{X}}{\vv{t}})\vv{v} \eval \vv{w}$
    if and only if
    $F(\vv{v}) = \vv{w}$.
\end{definition}

This definition, together with \Cref{lem:BasesIso}, allows us to build a
characterisation of unitary operators as values in
$\sem{\sharp\basis{X}\Arr\sharp\basis{Y}}$.

\begin{theorem}[Characterisation of Unitary Operators]
  Let $X$ and $Y$ be orthonormal bases of size~$n$.
  A closed $\lambda$-abstraction
  $\Lam{x}{{X}}{\vv{t}}\in\sem{\sharp\basis{X}\Arr\sharp\basis{Y}}$
  if and only if it represents a unitary operator
  $F:\C^n\to\C^n$.
\end{theorem}

\begin{proof}
  \textit{Necessity.}
  Suppose that $\Lam{x}{{X}}{\vv{t}}\in\sem{\sharp\basis{X}\Arr\sharp\basis{Y}}$.
  Then, by \Cref{lem:BasesIso}, for every
  $\vv{v_i}\in\sem{\basis{X}}$ there exists
  $\vv{w_i}\in\sem{\sharp\basis{Y}}$ such that
  $\vv{t}[\vv{v_i}/x]\eval\vv{w_i}$ and
  $\vv{w_i}\perp\vv{w_j}$ whenever $i\neq j$.
  Let $F:\C^n\to\C^n$ be the operator defined by
  $F(\vv{v_i})=\vv{w_i}$.
  By linearity over~$X$, it follows that
  $\Lam{x}{{X}}{\vv{t}}$ represents~$F$.
  Moreover, $F$ is unitary since
  $\|\vv{w_i}\|_{\C^n}=\|\vv{w_j}\|_{\C^n}=1$ and
  $\scal{\vv{w_i}}{\vv{w_j}}_{\C^n}=0$.

  \smallskip
  \textit{Sufficiency.}
  Conversely, suppose that $\Lam{x}{{X}}{\vv{t}}$ represents a
  unitary operator $F:\C^n\to\C^n$.
  Then, for each $\vv{v_i}\in\sem{\basis{X}}$, there exists
  $\vv{w_i}\in\sem{\basis{Y}}$ such that
  $F(\vv{v_i}) = \vv{w_i}$ and
  $(\Lam{x}{{X}}{\vv{t}})\vv{v_i}\eval\vv{w_i}$.
  Hence,
  \(
    (\Lam{x}{{X}}{\vv{t}})\vv{v_i}
    \lraneq
    \vv{t}\ansubst{\vv{v_i}/x}{X}
    = \vv{t}[\vv{v_i}/x]
    \eval \vv{w_i}
    \in \sem{\sharp\basis{Y}}
  \),
  since $\|\vv{w_i}\|=\|F(\vv{v_i})\|_{\C^n}=1$.
  From \Cref{lem:BasesIso}, it follows that
  $\Lam{x}{{X}}{\vv{t}}\in\sem{\sharp\basis{X}\Arr\sharp\basis{Y}}$,
  and moreover
  \(
    \scal{\vv{w_i}}{\vv{w_j}}
    = \scal{F(\vv{v_i})}{F(\vv{v_j})}_{\C^n}
    = 0.
  \)
  \qedhere
\end{proof}

  This result naturally extends to unitary distributions of
  $\lambda$-ab\-strac\-tions, since a term of the form
  $\Lam{x}{{X}}{\sum_{i=1}^{n}\alpha_i \vv{t_i}}$ is syntactically different
  but computationally equivalent to $\sum_{i=1}^{n}\alpha_i
  \Lam{x}{{X}}{\vv{t_i}}$.  Hence, the characterisation of unitary operators
  also applies to superpositions of abstractions sharing the same basis~$X$.

\subsection{Typing rules}    
In this section, we focus on enumerating and proving the validity of various
typing rules. The objective is to extract a reasonable set of rules that can
constitute a type system. We first need to lay the groundwork to properly define
what it means for a typing rule to be valid.

\begin{definition}
  A \emph{context} (denoted by capital Greek letters $\Gamma$, $\Delta$) is a
  finite mapping $\Gamma:\Var\to\Type\times\BasisType$ assigning a type and a
  basis to each variable in its domain. We write
  \(
    \Gamma = {x_1}^{{X_1}}:A_1,\dotsb, {x_n}^{{X_n}}:A_n
  \)
  to indicate that $\Gamma(x_i)=(A_i,\basis{X_i})$ for each~$i$.
\end{definition}

As in standard typing judgements, the context records the types of a term's
free variables. However, since substitution in our calculus depends on a basis,
we also wish to record that information. This is not strictly necessary, as the
basis with respect to which a variable is interpreted should not affect its
type. For instance, consider the following substitutions:
\begin{align*}
(\Lam{x}{\B}{\Pair{x}{y}})\ansubst{\ket{0}/y}{\B}
  &= \Lam{x}{\B}{\Pair{x}{\ket{0}}},\\
\text{and}\quad
(\Lam{x}{\B}{(x, y)})\ansubst{\ket{0}/y}{\XB}
  &= \tfrac{1}{\sqrt{2}}
    \bigl((\Lam{x}{\B}{\Pair{x}{\ket{+}}})
         + (\Lam{x}{\B}{\Pair{x}{\ket{-}}})\bigr).
\end{align*}
These terms are not syntactically identical, yet they are computationally
equivalent. 
Since typing via realisability captures computational behaviour,
their types coincide. Nevertheless, we retain basis information in contexts, as
it will simplify the forthcoming proofs.

\begin{definition}
  Given a context~$\Gamma$, its \emph{unitary semantics},
  denoted~$\sem{\Gamma}$, is the set of substitutions defined by
  \begin{align*}
    \sem{\Gamma}
    := 
    \{&\sigma~\text{substitution} \mid 
      \dom{\sigma} = \dom{\Gamma}
      \text{ and } \forall {x_i} \in \dom{\Gamma},\\
      &\Gamma(x_i) = (A_i, \basis{X_i})
      \text{ implies }
      \sigma(x_i) = \ansubst{\vv{v_i}/x_i}{{X_i}}
      \text{ for some }\vv{v_i} \in \sem{A_i}\}.
  \end{align*}
\end{definition}

To ensure a coherent treatment of quantum data, we must guarantee that qubits
are handled linearly. The first step is to identify which variables in the
context represent quantum data---those associated with a type of the
form~$\sharp A$. We call the subset of~$\Gamma$ composed of such variables its
\emph{strict domain}.

\begin{definition}
  The \emph{strict domain} of a context~$\Gamma$, denoted~$\sdom{\Gamma}$, is
  defined as
  \(
    \sdom{\Gamma} :=
    \{x \in \dom{\Gamma} \mid
    \sem{\Gamma(x)} = \sem{\sharp(\Gamma(x))}\}
  \).
\end{definition}

This definition relies on the idempotence of the~$\sharp$~operator
(\Cref{thm:IdempotentSharp}).

A typing judgement $\Gamma\vdash \vv{t}:A$ is valid if it satisfies two
conditions.  First, every free variable of~$\vv{t}$ must belong to the domain
of~$\Gamma$, and every variable in the strict domain~$\sdom{\Gamma}$ must occur
in~$\vv{t}$.  This ensures that no information is erased and that all variables
are properly accounted for.  The linear treatment of quantum data is thus
enforced by substitution.

Second, for every substitution in the unitary semantics of~$\Gamma$, applying
it to~$\vv{t}$ must yield a term that reduces to a realizer of type~$A$.  This
condition ensures that the operational behaviour of the term within the context
is faithfully captured by the type.  Formally:

\begin{definition}[Typing judgement]
  A \emph{typing judgement} $\TYP{\Gamma}{\vv{t}}{A}$ is valid when:
    $\sdom{\Gamma}\subseteq\FV{\vv{t}}\subseteq\dom{\Gamma}$
    and
    for all $\sigma\in\sem{\Gamma}$, 
          $\vv{t}\ansubst{\sigma}{}\real A$.
\end{definition}

We are also interested in \emph{orthogonal terms}, that is, terms that reduce
to orthogonal values.  
We therefore introduce the following notion.

\begin{definition}
  An \emph{orthogonality judgement}
  $\ORTH{\Gamma}{\Delta_1}{\vv{t}}{\Delta_2}{\vv{s}}{A}$
  is said to be \emph{valid} when
  \begin{itemize}
    \item 
  the judgements
      $\TYP{\Gamma,\Delta_1}{\vv{t}}{A}$ and
      $\TYP{\Gamma,\Delta_2}{\vv{s}}{A}$ are valid, and
    \item for
      for every
      $\sigma\in\sem{\Gamma,\Delta_1}$ and
      $\tau\in\sem{\Gamma,\Delta_2}$,
      there exist value distributions $\vv{v},\vv{w}$ such that
      $\vv{t}\ansubst{\sigma}{}\eval\vv{v}$,
      $\vv{s}\ansubst{\tau}{}\eval\vv{w}$,
      and $\vv{v}\perp\vv{w}$.
  \end{itemize}
When both $\Delta_1$ and $\Delta_2$ are empty,
we just write
$\SORTH{\Gamma}{\vv{t}}{\vv{s}}{A}$.
\end{definition}

With these definitions in mind, a typing rule is \emph{valid} when valid
premises entail a valid conclusion.  
Although there are infinitely many valid rules (each corresponding to a theorem),
\Cref{tab:TypingRules} presents a representative subset forming a reasonable
core typing system for the calculus, whose validity is stated below.

\begin{table}[t]
  \[
    \begin{array}{c}
      \infer[\snam{Axiom}]{\TYP{x^{X}:A}{x}{A}}{\basis{X}\leq A \text{ or }X=\AbsBasis}
      \qquad
	\infer[\snam{UnitLam}]{\TYP{\Gamma}{\sum_{i=1}^n \alpha_i (\Lam{x}{{X}}{\vv{t_i}})}{A\Arr B}}
	{\TYP{\Gamma,x^{X}:A}{\sum_{i=1}^{n}\alpha_i\vv{t_i}}{B}}
      \\
      \noalign{\medskip}
      \infer[\snam{App}]{\TYP{\Gamma,\Delta}{\vv{s}\,\vv{t}}{B}}
      {\TYP{\Gamma}{\vv{s}}{A\Arr B} & \TYP{\Delta}{\vv{t}}{A}}
      \qquad
      \infer[\snam{Pair}]{\TYP{\Gamma,\Delta}
      {\Pair{\vv{t}}{\vv{s}}}{A\times B}}{
	\TYP{\Gamma}{\vv{t}}{A}&\TYP{\Delta}{\vv{s}}{B}
      }
      \\
      \noalign{\medskip}
      \infer[\snam{LetPair}]{\TYP{\Gamma,\Delta} 
      {\LetP{x}{{X}}{y}{{Y}}{\vv{t}}{\vv{s}}}{C}}{
	\TYP{\Gamma}{\vv{t}}{A\times B}&
	\TYP{\Delta,x^{{X}}:A,y^{Y}:B}{\vv{s}}{C}
      }\\
      \noalign{\medskip}
      \infer[\snam{LetTens}]{\TYP{\Gamma,\Delta}
      {\LetP{x}{{X}}{y}{{Y}}{\vv{t}}{\vv{s}}}{\sharp C}}{
	\TYP{\Gamma}{\vv{t}}{\sharp(A\times B)}&
	\TYP{\Delta,x^{X}:\sharp A,y^{Y}:\sharp B}{\vv{s}}{C}
      }\\
      \noalign{\medskip}
      \infer[\snam{Case}]{\TYP{\Gamma,\Delta}
      {\gencase{\vv{t}}{\vv{v_1}}{\vv{v_n}}{\vv{s_1}}{\vv{s_n}}}{A}}{
	\TYP{\Gamma}{\vv{t}}{\genbasis{\vv{v_i}}{i=1}{n}}&
	\forall i,\ \TYP{\Delta}{\vv{s_i}}{A}
      }\\
      \noalign{\medskip}
      \infer[\snam{UnitCase}]{\TYP{\Gamma,\Delta}
      {\gencase{\vv{t}}{\vv{v_1}}{\vv{v_n}}{\vv{s_1}}{\vv{s_n}}}{\sharp A}}{
	\TYP{\Gamma}{\vv{t}}{\sharp \genbasis{\vv{v_i}}{i=1}{n}}&
	\forall i\neq j,\ \SORTH{\Delta}{\vv{s_i}}{\vv{s_j}}{A}
      }\\
      \noalign{\medskip}
      \infer[\snam{Sum}]
      {\TYP{\Gamma}{\sum_{i=1}^{n} \vv{t_i}}{\sharp A}}
      {\forall i\neq j,\, \SORTH{\Gamma}{\vv{t_i}}{\vv{t_j}}{A} &
      \sum_{i=1}^{n}|\alpha_i|^2 = 1}
      \\
      \noalign{\medskip}
      \infer[\snam{Contr}]{\TYP{\Gamma,x^{X}:{\basis{X}}}{\vv{t}\,[y:=x]}{B}}{
	\TYP{\Gamma,x^{X}:{\basis{X}},y^{X}:{\basis{X}}}{\vv{t}}{B}
      } 
      \qquad
	\infer[\snam{Weak}]{\TYP{\Gamma,x^{X}:\basis{X}}{\vv{t}}{C}}{
	  \TYP{\Gamma}{\vv{t}}{C}
	}
      \\
      \noalign{\medskip}
      \infer[\snam{Sub}]{\TYP{\Gamma}{\vv{t}}{B}}{\TYP{\Gamma}{\vv{t}}{A} & \SUB{A}{B}}
      \quad
      \infer[\snam{Equiv}]{\TYP{\Gamma}{\vv{s}}{A}}{
	\TYP{\Gamma}{\vv{t}}{A}& \vv t\equiv \vv s
      }
      \quad 
      \infer[\snam{Phase}]{\TYP{\Gamma}{e^{i\theta}\vv{t}}{A}}
      {\TYP{\Gamma}{\vv{t}}{A}}
    \end{array}
  \]
  \caption{Some valid typing rules}
  \label{tab:TypingRules}
\end{table}

\begin{theorem}\label{thm:TypingRulesValidity}
  All the typing rules in \Cref{tab:TypingRules} are valid.
  \qed
\end{theorem}

The usual safety properties follow straightforwardly in this framework.
\emph{Confluence} is an immediate consequence of the reduction being
deterministic (cf.~\Cref{rmk:determinism}).  
\emph{Strong normalisation} follows directly from the definition of a
realizer (cf.~\Cref{def:Realizer}).  
\emph{Subject reduction} is also immediate: indeed, if
$\TYP{\Gamma}{\vv{t}}{A}$ and $\vv{t}\to\vv{u}$, then
$\TYP{\Gamma}{\vv{u}}{A}$ by definition.
However, if we restrict ourselves to a subset of typing rules—such as those
presented in \Cref{tab:TypingRules}—we must
ensure that this restricted system still suffices to type all reducts of a
term, once the underlying realisability semantics is abstracted away.  In our
case, the rules proven valid in \Cref{thm:TypingRulesValidity}
suffice to guarantee subject reduction, as
stated below.

\begin{theorem}[Subject reduction]\label{thm:SubjectReduction}
  If $\TYP{\Gamma}{\vv{t}}{A}$ can be derived using the set of rules in
  \Cref{tab:TypingRules} and $\vv{t}\to\vv{u}$, then
  $\TYP{\Gamma}{\vv{u}}{A}$ can also be derived by the same set of rules.
  \qed
\end{theorem}

\section{Examples}\label{sec:examples}
\subsection{Deutsch's algorithm}\label{subsec:deutsch}
We begin with \emph{Deutsch's algorithm}, a canonical example that highlights
how basis types in the $\lambdaB$ calculus yield more informative typings for
quantum programs.
The algorithm is as follows.
We are given black-box access to an \emph{oracle}~$U_f$ that implements an
unknown Boolean function $f:\{0,1\}\to\{0,1\}$.  
The oracle can only be either \emph{constant} (both inputs map to the same
output) or \emph{balanced} (the two outputs differ).  
Classically, determining which case holds requires two queries to~$f$.  
Deutsch's algorithm decides this with a single query by exploiting quantum
superposition and interference.

Operationally, the oracle~$U_f$ acts as
\(
  U_f:\ket{xy}\mapsto \ket{x}\otimes\ket{y\oplus f(x)}
\),
where $\oplus$ is addition modulo~2.  
The textbook circuit prepares the state~$\ket{+-}$, applies~$U_f$, and
then applies a Hadamard on the first qubit before measuring it.  
The outcome is~$\ket{0}$ if~$f$ is constant and~$\ket{1}$ if~$f$ is balanced.

We begin with a standard implementation of this algorithm in $\lambdaB$.
We first encode the usual gates we will use:
\begin{align*}
  \Hd &:= \Lam{x}{\B}{\case{x}{\ket{0}}{\ket{1}}{\ket{+}}{\ket{-}}},\\
  \mathsf{NOT} &:= \Lam{x}{\B}{\case{x}{\ket{0}}{\ket{1}}{\ket{1}}{\ket{0}}},\\
  \mathsf{CNOT} &:= \Lam{x}{\B}{\Lam{y}{\B}{
    \case{x}{\ket{0}}{\ket{1}}
      {\Pair{\ket{0}}{y}}
      {\Pair{\ket{1}}{\pauliX{y}}}}}.
\end{align*}

We model the four possible oracles $U_f$ (two constant and two balanced):
\begin{align*}
  O_{\mathrm{const}\,0} &:= \Lam{x}{\B}{\Lam{y}{\B}{\Pair{x}{y}}},
  & O_{\mathrm{id}}       &:= \Lam{x}{\B}{\Lam{y}{\B}{\cnot{x}{y}}},\\
  O_{\mathrm{const}\,1} &:= \Lam{x}{\B}{\Lam{y}{\B}{\Pair{x}{(\pauliX{y})}}},
  & O_{\mathrm{flip}}     &:= \Lam{x}{\B}{\Lam{y}{\B}{\cnot{x}{(\pauliX{y})}}}.
\end{align*}

The standard Deutsch term prepares~$\ket{+-}$, calls the oracle, then
applies~$H$ on the first qubit and returns the (classical) first component:
\[
  \mathsf{Deutsch}_{\mathrm{std}} :=
  \Lam{f}{\AbsBasis}{
    \LetP{x}{\B}{y}{\B}
      {(f\,(\Hd\,\ket{0})\,(\Hd\,\ket{1}))}
      {\Pair{\Hd\,x}{y}}
  }.
\]

Each oracle above can be typed as
\(
  \basis{\B}\Arr\basis{\B}\Arr(\basis{\B}\times\basis{\B})
\),
but since the arguments we pass are $\ket{+}$ and $\ket{-}$ (superpositions),
the overall judgement we can derive for the application uses~$\sharp$:
\[
  \TYP{}{\mathsf{Deutsch}_{\mathrm{std}}}
  {(\sharp\basis{\B}\Arr\sharp\basis{\B}\Arr(\sharp\basis{\B}\times\sharp\basis{\B}))
   \Arr \sharp(\basis{\B}\times\basis{\B})}.
\]
This type is correct but coarse: it only guarantees that the result is a
unitary distribution of pairs of booleans.  Operationally we know more: the
first output is actually a basis bit ($\ket{0}$ or~$\ket{1}$) encoding whether
$f$ is constant or balanced, and the second output can be ignored.

Then, we can consider a basis-aware implementation with tighter typing.
The key observation is that the oracle is always called on the fixed state
$\ket{+-}$.  
Thus it is natural to write the program in the
$\XB=\{\ket{+},\ket{-}\}$ basis, letting the types track that we remain in a
basis state at the oracle boundary.

We define the same gates in the $\XB$ basis:
\begin{align*}
  \pauliZXB &:= \Lam{x}{\XB}{\case{x}{\ket{+}}{\ket{-}}{\ket{-}}{\ket{+}}},\\
  \pauliXXB &:= \Lam{x}{\XB}{\case{x}{\ket{+}}{\ket{-}}{\ket{+}}{-1\,\ket{-}}},\\
  \cnotXB &:= \Lam{x}{\XB}{\Lam{y}{\XB}{
    \case{y}{\ket{+}}{\ket{-}}
      {\Pair{x}{\ket{+}}}
      {\Pair{\pauliZXB{x}}{\ket{-}}}
  }}.
\end{align*}

We then rewrite the program and the four oracles in~$\XB$:
\begin{align*}
  \mathsf{Deutsch} &:= \Lam{f}{\AbsBasis}{ \LetP{x}{\XB}{y}{\XB}{(f\,\ket{+}\,\ket{-})} {\case{x}{\ket{+}}{\ket{-}}{\ket{0}}{\ket{1}}}},\\
\end{align*}
\vspace{-2\baselineskip}
\begin{align*}
  O_{\mathrm{const}\,0}^{\XB} &:= \Lam{x}{\XB}{\Lam{y}{\XB}{\Pair{x}{y}}},
  &O_{\mathrm{id}}^{\XB}\ \ &:= \Lam{x}{\XB}{\Lam{y}{\XB}{\cnotXB{x}{y}}},\\
  O_{\mathrm{const}\,1}^{\XB} &:= \Lam{x}{\XB}{\Lam{y}{\XB}{\Pair{x}{(\pauliXXB{y})}}},
  &O_{\mathrm{flip}}^{\XB} &:= \Lam{x}{\XB}{\Lam{y}{\XB}{\cnotXB{x}{(\pauliXXB{y})}}}.
\end{align*}

Now every oracle has the tight type
\(
  \basis{\XB}\Arr\basis{\XB}\Arr(\basis{\XB}\times\basis{\XB})
\),
and the program itself can be typed as
\[
  \TYP{}{\mathsf{Deutsch}}
  {(\basis{\XB}\Arr\basis{\XB}\Arr(\basis{\XB}\times\basis{\XB}))\Arr\basis{\B}}.
\]
Intuitively, the oracle---when fed with $\ket{+-}$---produces a pair of
basis states in~$\XB$ (up to a global phase).  Hence we can treat its output
classically: a single $\mathsf{case}$ on the first component suffices to return
a classical bit in the computational basis, with no need for a $\sharp$-type on
the result.

Both implementations are operationally equivalent: they compute a bit that
decides whether $f$ is constant or balanced.  The difference lies in the
\emph{precision} of their typings.  The standard version, expressed in~$\B$
with explicit Hadamards, forces $\sharp$ on the oracle's interface and thus
yields a result in~$\sharp(\basis{\B}\times\basis{\B})$.  The basis-aware
version states, via types, that the oracle is used on a fixed $\XB$-input and
therefore returns an $\XB$-basis pair; this lets us deterministically extract a
$\B$-basis bit.  In the typing derivation, no variable carries a $\sharp$-type:
(i) we may safely ignore the second qubit, and (ii) the first qubit is
guaranteed to be classical in~$\B$, reflecting the determinism of Deutsch's
algorithm.

\subsection{Quantum teleportation}\label{sec:teleportation}
We now turn to the \emph{quantum teleportation protocol}, a cornerstone of
quantum information theory that exemplifies the manipulation of entangled
states and the transmission of quantum data through classical communication.
Within the $\lambdaB$ calculus, this protocol provides a natural setting to
combine pattern matching, linear handling of qubits, and the
\emph{deferred-measurement principle}.  
Using the $\mathsf{case}$ constructor together with the expressive typing
discipline introduced earlier, we can encode the teleportation process in a
way that remains both syntactically compact and semantically faithful to its
quantum-mechanical counterpart.

The deferred-measurement principle states that any quantum circuit can delay
its measurements without altering the final outcome.  
More precisely, any gate classically controlled by the result of a measurement
is equivalent to another circuit where the control qubit remains unmeasured,
acting coherently on all branches of the superposition.  
Although the $\lambdaB$ calculus has no primitive operation for measurement,
the $\mathsf{case}$ constructor allows us to simulate such classically
controlled gates by branching on basis states.

A canonical example that exploits this principle is the
\emph{quantum teleportation protocol}.  
Two agents (Alice and Bob) share an entangled pair of qubits forming a Bell
state.  
Using this shared entanglement and two bits of classical information,
Alice can transmit an unknown quantum state~$\ket{\psi}$ to Bob without
physically sending the qubit itself.  
The standard circuit implementing the protocol is shown below:
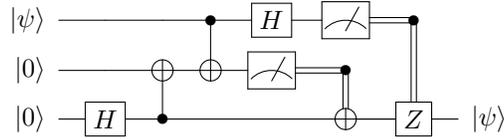
\[
  \Qcircuit @C=1em @R=.5em {
    \lstick{\ket{\psi}} & \qw & \qw & \ctrl{1} & \gate{H} & \meter & \control \cw \\
    \lstick{\ket{0}} & \qw & \targ & \targ & \meter & \control \cw \cwx[1] \\
    \lstick{\ket{0}} & \gate{H} & \ctrl{-1} & \qw & \qw & \targ & \gate{Z} \cwx[-2] \qw & \rstick{\ket{\psi}} \qw
  }
\]
The algorithm first creates the Bell state $\Phi^{+}$ between the second and
third qubits, then performs a Bell-basis measurement on the first and second
qubits.  
Operationally, this measurement is implemented by applying a CNOT gate followed
by a Hadamard on the first qubit (the adjoint of Bell-state preparation), and
then measuring both qubits.  
Depending on the pair of classical outcomes, a Pauli correction ($I$, $X$, $Z$,
or $ZX$) is applied to Bob's qubit to recover the original state~$\ket{\psi}$.

We can simulate the behaviour of this circuit in the $\lambdaB$ calculus by
defining a term that, instead of performing measurements, explicitly
describes the computation corresponding to each branch.  
A possible encoding is:
\begin{align*}
  &\mathsf{Teleport} :=
  \Lam{x}{\B}{
    \LetP{y_1}{\B}{y_2}{\B}{\Phi^{+}}{
      \mathsf{case}\;
        \Pair{x}{y_1}\;
        \mathsf{of}\\
         &\{ \Phi^{+} \mapsto \Pair{\Phi^{+}}{y_2},\;
          \Phi^{-} \mapsto \Pair{\Phi^{-}}{Z\,y_2},\;
	  \Psi^{+} \mapsto \Pair{\Psi^{+}}{X\,y_2},\;
	\Psi^{-} \mapsto \Pair{\Psi^{-}}{ZX\,y_2} \}
    }
  }.
\end{align*}

This term takes the input qubit~$\ket{\psi}$ and pairs it with one half of an
entangled Bell pair.  The $\mathsf{case}$ construct then matches the first two
qubits (the input and Alice's entangled qubit) against the Bell basis, and in
each branch applies the appropriate correction to Bob's qubit to
recover~$\ket{\psi}$.  The resulting term represents the same quantum
transformation as the circuit above but expressed without any explicit
measurement---instead, each branch encodes the coherent superposition of
possible measurement outcomes.

The $\lambdaB$ calculus allows us to abstract both the encoding and decoding
steps in the Bell basis, exploiting the deferred-measurement principle in a
type-safe way.  
Each branch of the $\mathsf{case}$ corresponds to a unitary transformation
preserving linearity and orthogonality, and the overall term has type
\[
  \TYP{}{\mathsf{Teleport}}
  {\sharp\basis{\B}\Arr(\sharp\basis{\Bell}\times\sharp\basis{\B})}.
\]
This type reflects that the protocol operates on superpositions of basis
states, producing a Bell-basis measurement outcome together with the recovered
qubit.  
In this way, the $\lambdaB$ calculus captures both the logical structure of
teleportation and its deferred-measurement semantics within a single, uniform
term language.

\section{Conclusion}\label{sec:conclusion}

In this paper we have explored a quantum-control $\lambda$-calculus equipped
with the additional feature of allowing abstractions to be expressed relative
to arbitrary bases, beyond the canonical one.  

The central mechanism enabling this extension is the decoration of
$\lambda$-ab\-strac\-tions and $\mathsf{let}$-constructors with basis annotations,
together with a modified substitution operation that governs how value
distributions decompose across different bases.  
These additions do not increase the expressive power of the original calculus
on which $\lambdaB$ builds, yet they offer a novel perspective for reasoning
about quantum programs and their behaviour under basis changes.

The reduction system coordinates computation through these extended syntactic
constructs and substitutions.  
A key property is that evaluation commutes with the congruence relation,
ensuring that interpreting a value distribution in a different basis does not
affect the computational result.  
Consequently, it is meaningful to reason about terms modulo basis
congruence.

The benefit of this design becomes clear in the realisability model.
The inclusion of atomic types~$\basis{X}$ enables a direct characterisation of
abstractions representing unitary operators---our main semantic result,
generalising the characterisation from~\cite{DiazcaroGuillermoMiquelValironLICS19}.  
Here, the use of basis types yields a simpler and more transparent proof.  

The second major result is the validity of the typing rules presented in
\Cref{tab:TypingRules}.  
By deriving these rules from the realisability interpretation, we ensure their
soundness and obtain a principled foundation for a typed programming language
based on the calculus.

Finally, we have illustrated the expressive advantages of the system through two
canonical examples.  
In the case of \emph{Deutsch's algorithm}, the use of basis-aware typing allows
the result to be treated classically, reflecting the algorithm's determinism.  
In the case of \emph{quantum teleportation}, we demonstrated how the
$\mathsf{case}$ construct can simulate gates controlled by Bell-basis
measurements, effectively capturing the deferred-measurement principle within
the calculus.

\begin{credits}
\subsubsection{\ackname} 
Supported by the European Union through the MSCA SE project QCOMICAL (Grant Agreement ID: 101182521) 
the Plan France 2030 through the PEPR integrated project EPiQ (ANR-21-PETQ-0007),
and by the Uruguayan CSIC grant 22520220100073UD.

\end{credits}
%
%
%
\bibliographystyle{splncs04}
\bibliography{basisSensitive}

\clearpage

\makeatletter
\renewcommand*\theHsection{\thesection}
\makeatother
\appendix

\section{Omitted proofs from Section~\ref{sec:calculus}}\label{sec:appendixA}

\begin{restatetheorem}[Restatement of \Cref{thm:UniqueDecomposition}]
  \itshape
  If $B$ is an $n$-dimensional basis, then every $n$-dimensional qubit has
  a unique decomposition over $B$.
\end{restatetheorem}
\begin{proof}
  Let $\vv{b_i}$ be the basis vectors of $B$. Suppose
  $\sum_{i=1}^n \alpha_i \vv{b_i}$ and $\sum_{i=1}^n \beta_i \vv{b_i}$
  are two decompositions of $\vv{v}$ over $B$. Then
  \[
    0=\vv{v}-\vv{v}=\sum_{i=1}^{n}(\alpha_i-\beta_i)\vv{b_i}.
  \]
  By linear independence, $\alpha_i=\beta_i$ for all $i$.
\end{proof}

\begin{restatecorollary}[Restatement of \Cref{cor:EquivalentDecomposition}]
  \itshape
  If $\vv{v}\equiv\vv{w}$, then they share the same decomposition over any
  basis $B$.
\end{restatecorollary}
\begin{proof}
  Since $\vv{v}-\vv{w}\equiv\vv{v}-\vv{v}\equiv\vv{w}-\vv{w}$, the same
  argument as in \Cref{thm:UniqueDecomposition} shows that $\vv{v}$ and $\vv{w}$ have the
  same decomposition over $B$.
\end{proof}

\begin{restatelemma}[Restatement of \Cref{lem:distributiveSubstitution}]
  \itshape
  For term distributions $\vv{t_i}$, a value distribution $\vv{v}$, a
  variable $x$, coefficients $\alpha_i\in\C$, and a basis $B$ such that
  $\ansubst{\vv v/x}{B}$ is defined:
  \[
    \Bigl(\sum_i \alpha_i\vv{t_i}\Bigr)\ansubst{\vv v/x}{B}
    \equiv
    \sum_i \alpha_i\vv{t_i} \ansubst{\vv v/x}{B}.
  \]
\end{restatelemma}
\begin{proof}
  Let $B\neq\AbsBasis$ and
  $\vv{v}\equiv\sum_{j=1}^n\beta_j\vv{b_j}$ with each $\vv{b_j}\in B$.
  Then
  \begin{align*}
    \Bigl(\sum_i \alpha_i\vv{t_i}\Bigr)\ansubst{\vv v/x}{B}
    &= \sum_{j=1}^n \beta_j\Bigl(\sum_{i=1}^{n}\alpha_i t_i\Bigr)[\vv{b_j}/x]\\
    &\equiv \sum_{i=1}^{n}\alpha_i\Bigl(\sum_{j=1}^n\beta_j t_i[\vv{b_j}/x]\Bigr)
    = \sum_{i=1}^{n}\alpha_i\vv{t_i}\ansubst{\vv v/x}{B}.
  \end{align*}
  The case $B=\AbsBasis$ is analogous.
\end{proof}

\begin{restatelemma}[Restatement of \Cref{lem:EquivSubstitutions}]
  \itshape
  For value distributions $\vv{v},\vv{w}$, a term distribution $\vv{t}$, and
  an orthonormal basis $B$ such that both
  $\ansubst{\vv{v}/x}{B}$ and $\ansubst{\vv{w}/x}{B}$ are defined:
  if $\vv{v}\equiv\vv{w}$, then
  $\vv{t}\ansubst{\vv{v}/x}{B}
  =\vv{t}\ansubst{\vv{w}/x}{B}$.
\end{restatelemma}
\begin{proof}
  Since $\vv{v}\equiv\vv{w}$, by
  \Cref{cor:EquivalentDecomposition},
  both can be written as
  $\vv{v}\equiv\vv{w}\equiv\sum_{i=1}^{n}\alpha_i\vv{b_i}$ with
  $\vv{b_i}\in B$. Hence
  \[
    \vv{t}\ansubst{\vv{v}/x}{B}
    = \sum_{i=1}^{n}\alpha_i\vv{t}[\vv{b_i}/x]
    = \vv{t}\ansubst{\vv{w}/x}{B}.
  \]
\end{proof}

\section{Omitted proofs from Section~\ref{sec:reduction}}\label{sec:appendixB}

\begin{lemma}[Weak diamond property for $\lraneq$]\label{lem:SquigDiamond}
  Let $\vv{t}, \vv{s_1}, \vv{s_2}$ term distributions such that $\vv{t}\lraneq{s_1}$ and $\vv{t}\lraneq\vv{s_2}$. Then, either there exists a term distribution $\vv{r}$ such that $\vv{s_1}\lraneq\vv{r}$ and $\vv{s_2}\lraneq \vv{r}$. Or, $\vv{s_1}=\vv{s_2}$. Diagrammatically:
  
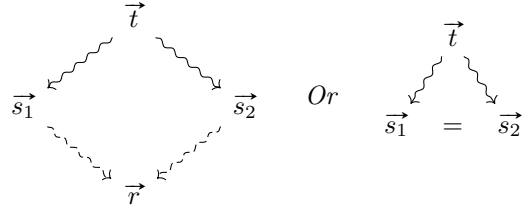
\[
    \begin{tikzcd}
      & \vv{t}
        \arrow[ld,decorate,decoration={snake, amplitude=0.8, segment length=6pt}, ->]
        \arrow[rd,decorate,decoration={snake, amplitude=0.8, segment length=6pt}, ->]
      &\\
      \vv{s_1}\arrow[dr,dashed,decorate,decoration={snake, amplitude=0.8, segment length=6pt}, ->] & &
      \vv{s_2}\arrow[ld,dashed,decorate,decoration={snake, amplitude=0.8, segment length=6pt}, ->] \\
      & \vv{r} &
    \end{tikzcd}
    \quad\text{ Or }\quad
    \begin{tikzcd}
      &[-2em] \vv{t}
        \arrow[ld,decorate,decoration={snake, amplitude=0.8, segment length=6pt}, ->]
        \arrow[rd,decorate,decoration={snake, amplitude=0.8, segment length=6pt}, ->]
      &[-2em]\\
      \vv{s_1}& = &\vv{s_2}\\
    \end{tikzcd}
  \]
\end{lemma}

\begin{proof}
  The proof follows from the fact that the $\lraneq$ reduction is deterministic over pure values. And, in case of term distributions, we only need to match the reduction on the corresponding sub-terms. Let $\vv{t}=\sum\limits_{i=1}^n \alpha_i \vv{t_i}$, $\vv{s_1}=\sum\limits_{i=1; i\neq j}^n \alpha_i \vv{t_i} + \alpha_j\vv{s_j}$, and $\vv{s_2}=\sum\limits_{i=1; i\neq k}^n \alpha_i \vv{t_i}+\alpha_k\vv{s_k}$. Where $\vv{t_j}\lraneq\vv{s_j}$ and $\vv{t_k}\lraneq\vv{s_k}$. If $j=k$ we are done, so we consider the case where $j\neq k$. Diagrammatically:
  
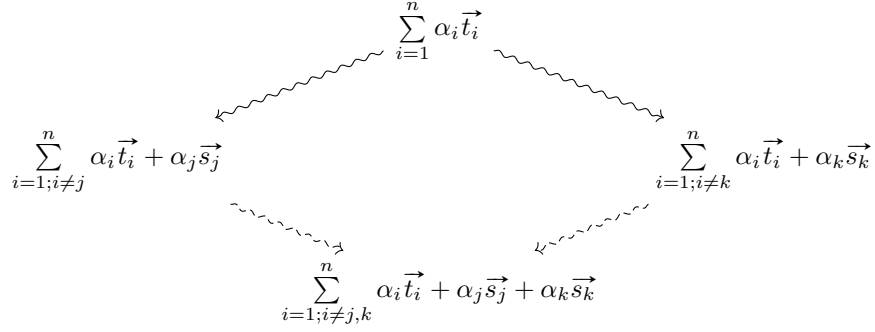
\[
    \begin{tikzcd}
      &[-1em] \sum\limits_{i=1}^n \alpha_i \vv{t_i}
        \arrow[ld,decorate,decoration={snake, amplitude=0.8, segment length=6pt}, ->]
        \arrow[rd,decorate,decoration={snake, amplitude=0.8, segment length=6pt}, ->]
      &[-1em]\\
      \sum\limits_{i=1; i\neq j}^n \alpha_i \vv{t_i} + \alpha_j\vv{s_j}
      \arrow[dr,dashed,decorate,decoration={snake, amplitude=0.8, segment length=6pt}, ->] & &
      \sum\limits_{i=1; i\neq k}^n \alpha_i \vv{t_i} + \alpha_k\vv{s_k}
      \arrow[ld,dashed,decorate,decoration={snake, amplitude=0.8, segment length=6pt}, ->] \\
      & \sum\limits_{i=1; i\neq j,k}^n \alpha_i \vv{t_i} + \alpha_j\vv{s_j}+\alpha_k\vv{s_k} &
    \end{tikzcd}
  \]
\end{proof}

\begin{restatetheorem}[Restatement of \Cref{thm:confluence}]
  \itshape
  Let $\vv{t}$ and $\vv{s}$ be closed term distributions with
  $\vv{t}\equiv\vv{s}$. If $\vv{t}\lraneq\vv{t'}$ and $\vv{s}\lraneq\vv{s'}$,
  then there exist term distributions $\vv{r_1}$ and $\vv{r_2}$ such that
  $\vv{t'}\eval\vv{r_1}$, $\vv{s'}\eval\vv{r_2}$, and
  $\vv{r_1}\equiv\vv{r_2}$.
  Diagrammatically:
  
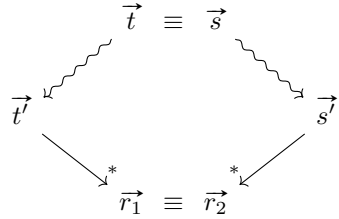
\[
    \begin{tikzcd}
      & \vv{t}
        \arrow[ld,decorate,decoration={snake, amplitude=0.8, segment length=6pt}, ->]
        &[-2.5em] \equiv
        &[-2.5em] \vv{s}
        \arrow[rd,decorate,decoration={snake, amplitude=0.8, segment length=6pt}, ->]
        &\\
      \vv{t'}\arrow[dr,"*",pos=0.9] & & & &
      \vv{s'}\arrow[ld,"*"',pos=0.9] \\
      & \vv{r_1} & \equiv & \vv{r_2} &
    \end{tikzcd}
  \]
\end{restatetheorem}
\begin{proof}
  We do a case-by-case analysis over the relation $\vv{t}\equiv\vv{s}$.
  \begin{description}
    \item[$\vv{t_1} + 0\vv{t_2}\equiv\vv{t_1}$:] This case follows from Lemma \ref{lem:SquigDiamond} since the reductions can only be performed in $\vv{t_1}$.
    
    \item[$0\vv{t}\equiv\vv{0}$:] The term distributions cannot reduce on either side of the equivalence.
    
    \item[$1\vv{t}\equiv\vv{t}$:] This case follows from Lemma \ref{lem:SquigDiamond}.
    
    \item[$\alpha(\beta \vv{t})\equiv\delta\vv{t}$:] This case follows from Lemma \ref{lem:SquigDiamond}.
    
    \item[$\vv{t_1}+\vv{t_2}\equiv\vv{t_2}+\vv{t_1}$:] This case follows from Lemma \ref{lem:SquigDiamond}. We just have to match the reductions on both sides of the equivalence.
    
    \item[$\vv{t_1}+(\vv{t_2}+\vv{t_3})\equiv(\vv{t_1}+\vv{t_2})+\vv{t_3}$:] This case follows from Lemma \ref{lem:SquigDiamond}. We just have to match the reductions on both sides of the equivalence.
    
    \item[$(\alpha+\beta)\vv{t}\equiv\vv{t}$:] We start analyzing the coefficients. If $\alpha+\beta = 0$, then there cannot be a reduction on the left hand-side. If $(\alpha + \beta)\neq 0$ and either $\alpha=0$ or $\beta=0$, then we are on a particular case of $\vv{t_1} + 0\vv{t_2}\equiv\vv{t_1}$ with $\vv{t_1}=\vv{t_2}$. Otherwise, we match the reductions on both sides of the equivalence with Lemma \ref{lem:SquigDiamond}.
    
    \item[$\alpha(\vv{t_1}+\vv{t_2})\equiv\alpha\vv{t_1}+\alpha\vv{t_2}$:] If $\alpha=0$, then the term distributions cannot reduce on either side of the equivalence. Otherwise, we match the reductions on both sides of the equivalence with Lemma \ref{lem:SquigDiamond}.
    
    \item[$\vv{t} (\alpha\vv{s})\equiv\alpha(\vv{t}\vv{s})$:] If $\alpha=0$, then there is no reduction possible on the right-hand side. If there is an internal reduction on either $\vv{s}$ or $\vv{t}$, then we match the reductions on both sides of the equivalence with Lemma \ref{lem:SquigDiamond}.
    
    If $\vv{t} = (\Lam{x}{B}{\vv{t_1}})$ and $\vv{s}=\vv{v}$ and $\vv{t_1}\ansubst{\vv{v}/x}{B}$, is defined then (we consider the case $B\neq\AbsBasis$):
    \begin{align*}
      (\Lam{x}{B}{\vv{t_i}}) (\alpha\vv{v}) &\lraneq \vv{t_1}\ansubst{\alpha\vv{v}/x}{B}\\
      &= \sum_{i=1}^{n} \alpha\beta_i \vv{t_1}[\vv{b_i}/x]\quad \text{with }\vv{v}\equiv\sum_{i=1}^{n} \beta_i \vv{b_i} \text{ with } \vv{b_i}\in B\\
    \end{align*}
    On the other side:
    \begin{align*}
      \alpha ((\Lam{x}{B}{\vv{t_1}}) \vv{v}) &\lraneq \alpha(\vv{t_1}\ansubst{\vv{v}}{B})\\
      &=\alpha(\sum_{i=1}^{n} \beta_i \vv{t_1}[\vv{b_i}/x])\quad \text{with }\vv{v}\equiv\sum_{i=1}^{n} \beta_i \vv{b_i} \text{ with } \vv{b_i}\in B\\
    \end{align*}

    And we have that both terms are equivalent. The case for $B=\AbsBasis$ is similar.

    \item[$(\alpha\vv{t})\vv{s}\equiv\alpha(\vv{t}\vv{s})$:] If $\alpha=0$, then there is no reduction possible on the right-hand side. If there is an internal reduction on either $\vv{s}$ or $\vv{t}$, then we match the reductions on both sides of the equivalence with Lemma \ref{lem:SquigDiamond}. There are no other possible redexes since the abstraction must be a pure value to reduce on the left hand-side.
    
    \item[$(\vv{t}+\vv{s})\vv{r}\equiv \vv{t}\vv{s} + \vv{t}\vv{r}$:] If there is an internal reduction on either $\vv{t}$,$\vv{s}$ or $\vv{r}$, then we match the reductions on both sides of the equivalence with Lemma \ref{lem:SquigDiamond}. There are no other possible redexes since the abstraction must be a pure value to reduce on the left hand-side.
    
    \item[$\vv{t}(\vv{s}+\vv{r})\equiv\vv{t}\vv{s} + \vv{t}\vv{r}$:] If there is an internal reduction on either $\vv{t}$, $\vv{s}$ or $\vv{r}$, then we match the reductions on both sides of the equivalence with Lemma \ref{lem:SquigDiamond}.
    
    If $\vv{t} = (\Lam{x}{B}{\vv{t_1}})$, $\vv{s}=\vv{v}$, and $\vv{r}=\vv{w}$ with
    $\vv{t_1}\ansubst{\vv{v}/x}{B}$ and $\vv{t_1}\ansubst{\vv{w}/x}{B}$ defined then (we consider the case where $B\neq\AbsBasis$):
    \begin{align*}
      (\Lam{x}{B}{\vv{t_1}}) (\vv{v}+\vv{w}) &\lraneq \vv{t_1}\ansubst{\vv{v}+\vv{w}/x}{B}\\
      &=\sum_{i=1}^{n}(\alpha_i+\beta_i)\vv{t_1}[\vv{b_i}/x]\\
      &\quad\text{where }\vv{v}\equiv\sum_{i=1}^{n}\alpha_i\vv{b_i},\vv{w}\equiv\sum_{i=1}^{n}\beta_i\vv{b_i}\text{ with }\vv{b_i}\in B 
    \end{align*}

    On the other side:
    \begin{align*}
      (\Lam{x}{B}{\vv{t_1}}) \vv{v} + (\Lam{x}{B}{\vv{t_1}}) \vv{w} &\lraneq \vv{t_1}\ansubst{\vv{v}/x}{B} + (\Lam{x}{B}{\vv{t_1}}) \vv{w}\\
      &\lraneq \vv{t_1}\ansubst{\vv{v}/x}{B} + \vv{t_1}\ansubst{\vv{w}/x}{B}\\
      &=\sum_{i=1}^{n}(\alpha_i)\vv{t_1}[\vv{b_i}/x] + =\sum_{i=1}^{n}(\alpha_i)\vv{t_1}[\vv{b_i}/x]\\
      &\quad\text{where }\vv{v}\equiv\sum_{i=1}^{n}\alpha_i\vv{b_i},\vv{w}\equiv\sum_{i=1}^{n}\beta_i\vv{b_i}\text{ with }\vv{b_i}\in B 
    \end{align*}

    And we have that both terms are equivalent. The case for $B=\AbsBasis$ is similar.

    \item[$\LetP{x_1}{B_1}{x_2}{B_2}{(\alpha \vv{t})}{\vv{s}}\equiv\alpha(\LetP{x_1}{A}{x_2}{B}{\vv{t}}{\vv{s}})$:] If $\alpha=0$, then there is no reduction possible on the right-hand side. If there is an internal reduction on either $\vv{s}$ or $\vv{t}$, then we match the reductions on both sides of the equivalence with Lemma \ref{lem:SquigDiamond}.
    
    If $\vv{t}=\vv{v}$ and $\vv{s}\ansubst{\vv{v}/x_1\otimes x_2}{B_1\otimes B_2}$ is defined, then (we consider $B_1,B_2\neq\AbsBasis$):
    \begin{align*}
    \LetP{x_1}{B_1}{x_2}{B_2}{(\alpha \vv{v})}{&\vv{s}}\lraneq \vv{s}\ansubst{\alpha\vv{v}/x_1\otimes x_2}{B_1\otimes B_2}\\
    &=\sum_{i=1}^{n}\alpha\beta_i \vv{s}[\vv{v_i}/x_1][\vv{w_i}/x_2]\\ 
    &\text{where: }\vv{v}\equiv\sum_{i=1}^n\beta_i\Pair{\vv{v_i}}{\vv{w_i}} \text{ with } \vv{v_i}\in B_1, \vv{w_i}\in B_2
    \end{align*}
    On the other side;
    \begin{align*}
    \alpha(\LetP{x_1}{B_1}{x_2}{B_2}{\vv{v}}{&\vv{s}})\lraneq \alpha(\vv{s}\ansubst{\vv{v}/x_1\otimes x_2}{B_1\otimes B_2})\\
    &=\alpha\sum_{i=1}^{n}\beta_i \vv{s}[\vv{v_i}/x_1][\vv{w_i}/x_2]\\ 
    &\text{where: } \vv{v}\equiv\sum_{i=1}^n\beta_i\Pair{\vv{v_i}}{\vv{w_i}} \text{ with } \vv{v_i}\in B_1, \vv{w_i}\in B_2
    \end{align*}

    And we have that both terms are equivalent. The case for $B_1,B_2=\AbsBasis$ are similar.

    \item[\parbox{.75\linewidth}{\begin{align*}
      &\LetP{x_1}{B_1}{x_2}{B_2}{\vv{t}+\vv{s}}{\vv{r}}\equiv\\
      &(\LetP{x_1}{B_1}{x_2}{B_2}{\vv{t}}{\vv{r}}) +
      (\LetP{x_1}{B_2}{x_2}{B_2}{\vv{s}}{\vv{r}})
      \end{align*}}:]\hfill\\
      If there is an internal reduction on either $\vv{t}, \vv{s}$ or $\vv{r}$, then we match the reductions on both sides of the equivalence with Lemma \ref{lem:SquigDiamond}.

      If $\vv{t}=\vv{v}$ and $\vv{s}=\vv{w}$ with $\vv{r}\ansubst{\vv{v}/x_1\otimes x_2}{B_1\otimes B_2}$ and $\vv{r}\ansubst{\vv{v}/x_1\otimes x_2}{B_1\otimes B_2}$ defined, then (we consider $B_1,B_2\neq\AbsBasis$):
      \begin{align*}
      \LetP{x_1}{B_1}{x_2}{B_2}{\vv{v}+\vv{w}}{\vv{r}}&\lraneq\vv{r}\ansubst{\vv{v}+\vv{w}/x_1\otimes x_2}{B_1\otimes B_2}\\
      &=\sum_{i=1}^n (\alpha_i+\beta_i) \vv{r}[\vv{v_i}/x_1][\vv{w_i}/x_2]\\
      &\text{where: }\vv{v}\equiv\sum_{i=1}^n\alpha_i\Pair{\vv{v_i}}{\vv{w_i}} \text{ with } \vv{v_i}\in B_1, \vv{w_i}\in B_2\\
      &\text{and: }\vv{w}\equiv\sum_{i=1}^n\beta_i\Pair{\vv{v_i}}{\vv{w_i}} \text{ with } \vv{v_i}\in B_1, \vv{w_i}\in B_2\\
      \end{align*}

      On the other side:
      \begin{align*}
      &(\LetP{x_1}{B_1}{x_2}{B_2}{\vv{t}}{\vv{r}}) + (\LetP{x_1}{B_2}{x_2}{B_2}{\vv{s}}{\vv{r}})\\
      &\lraneq\vv{r}\ansubst{\vv{v}/x_1\otimes x_2}{B_1\otimes B_2} + (\LetP{x_1}{B_2}{x_2}{B_2}{\vv{s}}{\vv{r}})\\
      &\lraneq\vv{r}\ansubst{\vv{v}/x_1\otimes x_2}{B_1\otimes B_2} + \vv{r}\ansubst{\vv{w}/x_1\otimes x_2}{B_1\otimes B_2}\\
      &=\sum_{i=1}^n\alpha_i\vv{r}[\vv{v_i}/x_1][\vv{w_i}/x_2] + \sum_{i=1}^n \beta_i \vv{r}[\vv{v_i}/x_1][\vv{w_i}/x_2]\\
      &\text{where: }\vv{v}\equiv\sum_{i=1}^n\alpha_i\Pair{\vv{v_i}}{\vv{w_i}} \text{ with } \vv{v_i}\in B_1, \vv{w_i}\in B_2\\
      &\text{and: }\vv{w}\equiv\sum_{i=1}^n\beta_i\Pair{\vv{v_i}}{\vv{w_i}} \text{ with } \vv{v_i}\in B_1, \vv{w_i}\in B_2\\
      \end{align*}
        
      And we have that both terms are equivalent. The case for $B_1,B_2=\AbsBasis$ are similar.

    \item[\parbox{.55\linewidth}{\begin{align*}
      &\gencase{\alpha \vv{t}}{\vv{v_1}}{\vv{v_n}}{\vv{s_1}}{\vv{s_n}}\equiv\\
      &\alpha(\gencase{\vv{t}}{\vv{v_1}}{\vv{v_n}}{\vv{s_1}}{\vv{s_n}})
    \end{align*}}:] \hfill\\
    
    If $\alpha=0$, then there is no reduction possible on the right hand-side. If there are internal reductions on $\vv{t}$, then we match on both sides of the equivalence with Lemma \ref{lem:SquigDiamond}.
    
    If $\vv{t}=\vv{v}\equiv\sum_{i=1}^n\beta_i \vv{v_i}$. Then:
    \begin{align*}
      \gencase{\alpha \vv{t}}{\vv{v_1}}{\vv{v_n}}{\vv{s_1}}{\vv{s_n}}\lraneq \sum_{i=1}^n\alpha\beta_i \vv{s_i}
    \end{align*}
    On the other side:
    \begin{align*}
      \alpha(\gencase{\vv{t}}{\vv{v_1}}{\vv{v_n}}{\vv{s_1}}{\vv{s_n}})\lraneq \alpha\sum_{i=1}^n\beta_i \vv{s_i}
    \end{align*}
    And we have that both terms are equivalent.

    \item[\parbox{.55\linewidth}{\begin{align*}
      &\gencase{(\vv{t}+\vv{s})}{\vv{v_1}}{\vv{v_n}}{\vv{r_1}}{\vv{r_n}}\equiv\\ 
      &\gencase{\vv{t}}{\vv{v_1}}{\vv{v_n}}{\vv{r_1}}{\vv{r_n}}+\\
      &\gencase{\vv{s}}{\vv{v_1}}{\vv{v_n}}{\vv{r_1}}{\vv{r_n}}  
    \end{align*}}:]\hfill\\

    If there is an internal reduction on either $\vv{t}$ or $\vv{s}$, then we match the reductions on both sides of the equivalence with Lemma \ref{lem:SquigDiamond}.

    If $\vv{t}=\vv{v}\equiv\sum_{i=1}^n\alpha_i \vv{v_i}$, and $\vv{s}=\vv{w}\equiv\sum_{i=1}^n\beta_i \vv{v_i}$. Then:
    \begin{align*}
      \gencase{(\vv{t}+\vv{s})}{\vv{v_1}}{\vv{v_n}}{\vv{r_1}}{\vv{r_n}}\lraneq \sum_{i=1}^n (\alpha_i+\beta_i) \vv{r_i}
    \end{align*}
    On the other side:
    \begin{align*}
      &\gencase{\vv{t}}{\vv{v_1}}{\vv{v_n}}{\vv{r_1}}{\vv{r_n}}+\gencase{\vv{s}}{\vv{v_1}}{\vv{v_n}}{\vv{r_1}}{\vv{r_n}}\\
      &\lraneq \sum_{i=1}^n \alpha_i+ \vv{r_i} + \gencase{\vv{s}}{\vv{v_1}}{\vv{v_n}}{\vv{r_1}}{\vv{r_n}}\\
      &\lraneq \sum_{i=1}^n \alpha_i \vv{r_i} + \sum_{i=1}^n \beta_i \vv{r_i}\\
    \end{align*}

    And we have that both terms are equivalent.
  \end{description}
  
\end{proof}

\section{Omitted proofs from Section~\ref{sec:model}}\label{sec:appendixC}

\begin{restatetheorem}[Restatement of \Cref{thm:SharpCharacterization}]
  \itshape
  The interpretation of a type~$\sharp A$ contains precisely the
  norm-$1$ linear combinations of values in~$\sem{A}$:
  \[
    \sem{\sharp A}
    = (\sem{A}^\bot)^\bot
    = \Span(\sem{A}) \cap \Sph.
  \]
\end{restatetheorem}
\begin{proof}
  Proof by double inclusion.
  \begin{description}
    \item[$\Span(\sem{A})\cap\Sph\subseteq (\sem{A}^\bot)^\bot$:] Let $\vv{v}\in\Span(\sem{A})\cap\Sph$. Then $\vv{v}$ is of the form $\sum_{i=1}^{n}\alpha_i \vv{v_i}$ with $\vv{v_i}\in\sem{A}$. Taking $\vv{w}\in\sem{A}^\bot$, we examine the inner product:
    
    \begin{align*}
    \scal{\vv{v}}{\vv{w}} &= \scal{\sum_{i=1}^{n}\alpha_i \vv{v_i}}{\vv{w}}\\
    &= \sum_{i=1}^{n}\overline{\alpha_i}\scal{\vv{v_i}}{\vv{w}}=0
    \end{align*}

    Then $\vv{v}\in(\sem{A}^\bot)^\bot$.

    \item[$(\sem{A}^\bot)^\bot\subseteq \Span(\sem{A})\cap\Sph$:] Reasoning by contradiction, we assume that there is a $\vv{v}\in(\sem{A}^\bot)^\bot$ such that $v\not\in\Span(\sem{A})\cap\Sph$. Since $\vv{v}\not\in\Span(\sem{A})$, $\vv{v}=\vv{w_1} + \vv{w_2}$ such that $\vv{w_1}\in\Span{\sem{A}}$ and $\vv{w_2}$ is a non-null vector which cannot be written as a linear combination of elements of $\sem{A}$. In other words, $\vv{w_2}\in\sem{A}^\bot$. Taking the inner product:
    \[
    \scal{\vv{v}}{\vv{w_2}} = \scal{\vv{w_1}+\vv{w_2}}{\vv{w_2}} = \|\vv{w_2}\|\neq 0
    \]
    Then $\vv{v}\not\in(\sem{A}^\bot)^\bot$. The contradiction stems from assuming $\vv{v}\not\in\Span{\sem{A}}\cap\Sph$.\qedhere
  \end{description}
\end{proof}

\begin{restatetheorem}[Restatement of \Cref{thm:IdempotentSharp}]
  \itshape
  The~$\sharp$ operator is idempotent; that is,
  $\sem{\sharp A} = \sem{\sharp(\sharp A)}$.
\end{restatetheorem}
\begin{proof}
  We want to prove that $(((\comp{\sem{A}})^\bot)^\bot)^\bot = (\comp{\sem{A}})^\bot$. For ease of reading, we will write $\comp[n]{A}$ for $n$ successive applications of the operation $\bot$.

  \begin{description}
    \item[$A\subseteq A^{\bot^2}$:] Let $\vv{v}\in A$. Then, for all $\vv{w}\in\comp{A}$, $\scal{\vv{v}}{\vv{w}} = 0$. Then $\vv{v}\in\comp[2]{A}$. With this we have $A\subseteq\comp[2]{A}$.
    
    \item[$A^{\bot^3}\subseteq \comp{A}$:] Let $\vv{u}\in \comp[3]{A}$. Then, for all $\vv{v}\in\comp[2]{A}$, $\scal{\vv u}{\vv v} = 0$. Since we have shown that $A\subseteq \comp[2]{A}$, we have that for all $\vv{w}\in A$, $\scal{\vv u}{\vv w} = 0$. Then $\vv u\in\comp{A}$. With this we have $\comp[3]{A}\subseteq \comp{A}$.
  \end{description}

  With these two inclusions we have that $\comp{A}=\comp[3]{A}$. So we conclude that: $\sem{\sharp(\sharp A)} = \comp[4]{A} = \comp[2]{A} = \sem{\sharp A}$ \qedhere
\end{proof}

\begin{restatetheorem}[Restatement of \Cref{prop:UnitaryTypes}]
  \itshape
  For every type~$A$, $\sem{A}\subseteq\Sph$.
\end{restatetheorem}

\begin{proof}
  Proof by induction on the shape of $A$. Since by definition, $\sem{\basis{X}}$, $\sem{A\Arr B}$ and $\sem{\sharp{A}}$ are built from values in $\Sph$ the only case we need to examine is $\sem{A\times B}$.
  
  Let $\vv v = \sum_{i=0}^{n} \alpha_i v_i \in\sem{A}$ and $\vv w = \sum_{j=0}^{m} \beta_j w_j$ where every $v_i$ are pairwise orthogonal, same for $w_j$. Then:
     
  \[(\vv v, \vv w) = \sum_{i=0}^{n} \sum_{j=0}^{m} \alpha_i\beta_j (v_i,w_j)\]
  
  So we have: 
  \[\|\Pair{\vv v}{\vv w}\| = \sqrt{\sum_{i=1}^n\sum_{j=1}^{m} |\alpha_i\beta_j|^2} = \sqrt{\sum_{i=1}^n |\alpha_i|^2 \sum_{j=1}^{m} |\beta_j|^2}\]

  Since both $\vv v\in\sem{A}$ and $\vv w\in\sem{B}$, by inductive hypothesis, we have that $\|\vv v\| = \| \vv w \| = 1$. Which is to say $\sum_{i=1}^{n} |\alpha_i|^2 = \sum_{j=1}^{m} |\beta_j| = 1$. So we conclude $\|\Pair{\vv{v}}{\vv{w}}\| = 1$.
  \qed
\end{proof}

\begin{restatelemma}[Restatement of \Cref{lem:BasesIso}]
  Let $X$ and $Y$ be orthonormal bases of the same finite
  dimension, and let $\Lam{x}{{X}}{\vv t}$ be a closed $\lambda$-abstraction.
  Then $\Lam{x}{{X}}{\vv t}\in\sem{\sharp\basis{X}\Arr\sharp\basis{Y}}$
  if and only if 
  for all $\vv{v_i},\vv{v_j}\in\sem{\basis{X}}$,
  there exist value distributions
  $\vv{w_i},\vv{w_j}\in\sem{\sharp\basis{Y}}$ such that,
  \[
    \vv{t}[\vv{v_i}/x]\eval\vv{w_i}
    \quad\text{and}\quad
    \vv{t}[\vv{v_j}/x]\eval\vv{w_j},
    \quad\text{with } 
    \vv{w_i}\perp\vv{w_j}\text{ whenever }i\neq j.
  \]
\end{restatelemma}
\begin{proof}
  \textit{The condition is necessary:} Suppose that $\Lam{x}{{X}}{\vv{t_k}}\in\sem{\sharp\basis{X}\Arr\sharp\basis{Y}}$, thus $\forall \vv{v_i}\in\sem{\sharp\basis{X}},\ \vv{t}\ansubst{\vv{v_i}/x}{X}\eval\vv{w_i}\in\sem{\sharp\basis{Y}}$. It remains to be seen that $\vv{w_i} \perp \vv{w_j}$ if $i\neq j$. For that, we consider $\alpha_i\in\C$ such that $\sum_{i=1}^n |\alpha_i|^2 = 1$. By linear application on the basis $X$ we observe that:
  \begin{align*}
    (\Lam{x}{{X}}{\vv{t}})(\sum_{i=1}^n \alpha_i \vv{v_i}) &\lra \vv t\ansubst{\sum_{i=1}^n \alpha_i \vv{v_i}/x}{X}
    = \sum_{i=1}^{n} \alpha_i \vv{t}[\vv{v_i}/x] 
    \eval \sum_{i=1}^n \alpha_i \vv{w_i}
  \end{align*}

  But since $\sum_{i=1}^n \alpha_i \vv{v_i}\in\sem{\sharp A}$, then $\sum_{i=1}^n \alpha_i \vv{w_i}\in\sem{\sharp B}$ too. Which implies $\|\sum_{i=1}^n \alpha_i \vv{w_i}\|=1$. Therefore:
  \begin{align*}
    1 = \|\sum_{i=1}^n \alpha_i \vv{w_i}\| &= \scal{\sum_{i=1}^n \alpha_i \vv{w_i}}{\sum_{j=1}^n \alpha_j \vv{w_j}}\\
    &=\sum_{i=1}^n |\alpha_i|^2 \scal{\vv{w_i}}{\vv{w_i} } + \sum_{i,j=1; i\neq j}^n \bar{\alpha_i}\alpha_j \scal{\vv{w_i}}{\vv{w_j}}\\
    &=\sum_{i=1}^n |\alpha_i|^2 \scal{\vv{w_i}}{\vv{w_i} } + \sum_{i,j=1; i<j}^n 2~\Rpart{\bar{\alpha_i}\alpha_j \scal{\vv{w_i}}{\vv{w_j}}}\\
    &=\sum_{i=1}^n |\alpha_i|^2 \|\vv{w_i}\|^2 + 2\sum_{i,j=1; i<j}^n \Rpart{\bar{\alpha_i}\alpha_j \scal{\vv{w_i}}{\vv{w_j}}}\\
    &=\sum_{i=1}^n |\alpha_i|^2 + 2\sum_{i,j=1; i<j}^n\Rpart{\bar{\alpha_i}\alpha_j \scal{\vv{w_i}}{\vv{w_j}}}\\
    &= 1 + 2\sum_{i,j=1; i<j}^n \Rpart{\bar{\alpha_i}\alpha_j \scal{\vv{w_i}}{\vv{w_j}}}
  \end{align*}

  And thus we are left with $\sum_{i,j=1; i<j}^n \Rpart{\bar{\alpha_i}\alpha_j \scal{\vv{w_i}}{\vv{w_j}}} = 0$. Taking $\alpha_{i'} = \alpha_{j'} = \frac{1}{\sqrt{2}}$ with $0$ for the rest of coefficients, we have $\Rpart{\scal{\vv{w_{i'}}}{\vv{w_{j'}}}} = 0$ for any two arbitrary $i'$ and $j'$. In the same way, taking $\alpha_{i'} = \frac{1}{\sqrt{2}}$ and $\alpha_{j'}=\frac{i}{\sqrt{2}}$ with $0$ for the rest of the coefficients, we have $\Ipart{\scal{\vv{w_{i'}}}{\vv{w_{j'}}}} = 0$ for any two arbitrary $i'$ and $j'$. Finally, we can conclude that $\scal{\vv{w_i}}{\vv{w_j}}=0$ if $i\neq j$.

  \textit{The condition is sufficient:} Suppose that there are $\vv{w_i}\in\sem{\sharp\basis{Y}}$ such that for every $\vv{v_i}\in\sem{\basis{X}}$:
  \[
    \vv t[\vv{v_i}/x] \eval \vv{w_i} \perp \vv{w_j} \lave \vv t[\vv{v_j}/x]\qquad \text{If } i\neq j
  \]
  Given any $\vv u\in\sem{\sharp\basis{X}}$ we have that $\vv u = \sum_{i=1}^n \alpha_i \vv{v_i}$ with $\sum_{i=1}^n |\alpha_i|^2 = 1$ and $\vv{v_i}\in\sem{\basis{X}}$. Then 
  \[
    (\Lam{x}{{X}}{\vv t}) \vv u \lra \vv{t_k}\ansubst{\vv u/x}{X}=\sum_{i=1}^{n}\alpha_i \vv{t}[\vv{v_i}/x]\eval\sum_{i=1}^n \alpha_i\vv{w_i}
  \]

  We have that for each $i$, $\vv{w_i}\in\sem{\sharp\basis{Y}}$. In order to show that $(\Lam{x}{A}{\vv t})\vv u\real\sharp\basis{Y}$ we still have to prove that $\|\sum_{i=1}^n \alpha_i \vv{w_i}\| = 1$

  \begin{align*}
    \|\sum_{i=1}^n \alpha_i \vv{w_i}\|^2 &= \scal{\sum_{i=1}^n \alpha_i \vv{w_i}}{\sum_{j=1}^n \alpha_j \vv{w_j}}\\
    &=\sum_{i=1}^n |\alpha_i|^2 \scal{\vv{w_i}}{\vv{w_i} } + \sum_{i,j=1; i\neq j}^n \bar{\alpha_i}\alpha_j \scal{\vv{w_i}}{\vv{w_j}}\\
    &=\sum_{i=1}^n |\alpha_i|^2 + 0\\
    &= 1
  \end{align*}

  Then $\sum_{i=1}^n \alpha_i \vv{w_i}\in\sem{\sharp(\sharp\basis{Y})}=\sem{\sharp\basis{Y}}$ by \Cref{thm:IdempotentSharp}. Since for every $\vv u\in\sem{\sharp A}$, $(\Lam{x}{A}{\vv t}) \vv u\real\sharp B$, we can conclude that $\Lam{x}{A}{\vv t}\in\sem{\sharp A\Arr\sharp B}$.\qedhere
\end{proof}

Before proving the soundness of the typing rules (\Cref{thm:TypingRulesValidity}), we need the following results.

\begin{theorem}\label{prop:InnerProdPairs} For all value distributions $\vv{v_1}, \vv{v_2}, \vv{w_1}, \vv{w_2}$ we have:
\[
\scal{\Pair{\vv{v_1}}{\vv{w_1}}}{\Pair{\vv{v_2}}{\vv{w_2}}} = \scal{\vv{v_1}}{\vv{v_2}}\scal{\vv{w_1}}{\vv{w_2}}
\]
\begin{proof}
    Let us write $\vv{v_1}=\sum_{i_1=1}^{n_1}\alpha_{i_1} v_{i_1}$, $\vv{v_2}=\sum_{i_2=1}^{n_2}\alpha'_{i_2} v_{i_2}$, $\vv{w_1}=\sum_{j_1=1}^{m_1}\beta_{j_1} w_{j_1}$ and $\vv{w_2}=\sum_{j_2=1}^{m_2}\beta'_{j_2} w_{j_2}$. Then we have:
    \begin{align*}
        &\scal{\Pair{\vv{v_1}}{\vv{w_1}}}{\Pair{\vv{v_2}}{\vv{w_2}}}\\
        &=\scal{\sum_{i_1=1}^{n_1}\sum_{j_1=1}^{m_1} \alpha_{i_1}\beta'_{j_1}\Pair{v_{i_1}}{w_{j_1}}}{\sum_{i_2=1}^{n_2}\sum_{j_2=1}^{m_2} \alpha_{i_2}\beta'_{j_2}\Pair{v_{i_2}}{w_{j_2}}}\\
        &=\sum_{i_1}^{n_1}\sum_{j_1}^{m_1}\sum_{i_2}^{n_2}\sum_{j_2}^{m_2} \overline{\alpha_{i_1}\beta_{j_1}} \alpha'_{i_2}\beta'_{j_2} \scal{\Pair{v_{i_1}}{w_{j_1}}}{\Pair{v_{i_2}}{w_{j_2}}}\\
        &=\sum_{i_1}^{n_1}\sum_{j_1}^{m_1}\sum_{i_2}^{n_2}\sum_{j_2}^{m_2} \overline{\alpha_{i_1}\beta_{j_1}} \alpha'_{i_2}\beta'_{j_2} \Kron{\Pair{v_{i_1}}{w_{j_1}}}{\Pair{v_{i_2}}{w_{j_2}}}\\
        &=\sum_{i_1}^{n_1}\sum_{j_1}^{m_1}\sum_{i_2}^{n_2}\sum_{j_2}^{m_2} \overline{\alpha_{i_1}\beta_{j_1}} \alpha'_{i_2}\beta'_{j_2} \Kron{v_{i_1}}{v_{i_2}}\Kron{w_{j_1}}{w_{j_2}}\\
        &=(\sum_{i_1}^{n_1}\sum_{j_1}^{m_1}\overline{\alpha_{i_1}}\alpha'_{i_2}\Kron{v_{i_1}}{v_{i_2}})(\sum_{i_2}^{n_2}\sum_{j_2}^{m_2} \overline{\beta_{j_1}} \beta'_{j_2} \Kron{w_{j_1}}{w_{j_2}})\\
        &=(\sum_{i_1}^{n_1}\sum_{j_1}^{m_1}\overline{\alpha_{i_1}}\alpha'_{i_2}\Pair{v_{i_1}}{v_{i_2}})(\sum_{i_2}^{n_2}\sum_{j_2}^{m_2} \overline{\beta_{j_1}} \beta'_{j_2} \Pair{w_{j_1}}{w_{j_2}})\\
        &=\scal{\vv{v_1}}{\vv{v_2}}\scal{\vv{w_1}}{\vv{w_2}}\\
    \end{align*}
\end{proof}  

\end{theorem}

\begin{lemma}\label{lem:VecRewrite}
Given a type $A$, two vectors $\vv{u_1},\vv{u_2}\in\sem{\sharp A}$ and a scalar $\alpha\in\C$, there exists a vector $\vv{u_0}\in\sem{\sharp A}$ and a scalar $\lambda\in\C$ such that:
\[
\vv{u_1} + \alpha\vv{u_2} = \lambda \vv{u_0} 
\]
\end{lemma}
\begin{proof}
    Let $\lambda:=\|\vv{u_1}+\alpha\vv{u_2}\|$. When $\lambda\neq 0$, we take $\vv{u_0}=\frac{1}{\lambda}(\vv{u_1}+\alpha\vv{u_2})\in\sem{\sharp A}$, and we are done.

    When $\lambda=0$, we first observe that $\alpha\neq 0$ since it would mean that $\|\vv{u_1}\|=0$ which is absurd since $\|\vv{u_1}\|=1$. Moreover, since $\lambda=\|\vv{u_1}+\alpha\vv{u_2}\|=0$, we observe that all the coefficients of the distribution $\vv{u_1}+\alpha\vv{u_2}$ are zeroes when written in canonical form which implies that:
    \[
    \vv{u_1}+\alpha\vv{u_2} = 0(\vv{u_1}+\alpha\vv{u_2}) = 0\vv{u_1}+0\vv{u_2}
    \]
    Using the triangular inequality we observe that:
    \begin{align*}
    0 &< 2|\alpha|\\
    &= \|2\alpha\vv{u_2}\|\\
    &\leq\|\vv{u_1}+\alpha\vv{u_2}\| + \|\vv{u_1 }+ (-\alpha)\vv{u_2}\|\\
    &= \|\vv{u_1}+(-\alpha)\vv{u_2}\|
    \end{align*}
    Hence $\lambda' := \|\vv{u_1}+(-\alpha)\vv{u_2}\|>0$. Taking $\vv{u_0}:= \frac{1}{\lambda'}(\vv{u_1}+ (-\alpha)\vv{u_2})\in\sem{\sharp A}$, we easily see that:
    \[
    \vv{u_1}+\alpha\vv{u_2} = 0\vv{u_1} + 0\vv{u_2} = 0(\frac{1}{\lambda'} (\vv{u_1} + (-\alpha) \vv{u_2})) = \lambda \vv{u_0}
    \]
\end{proof}

\begin{theorem}[Polarization identity]\label{prop:Polarization} 
For all values $\vv{v}$ and $\vv{w}$ we have:
\[
  \scal{\vv{v}}{\vv{w}}=
  \frac{1}{4} (\|\vv{v}+\vv{w}\|^2 - \|\vv{v} + (-1) \vv{w}\|^2 - i\|\vv{v} + i\vv{w}\|^2 + i\|\vv{v}+ (-i)\vv{w}\|^2)
\]
\end{theorem}

\begin{lemma}\label{lem:InnerProdSingleVar} 
Given a valid typing judgement of the term $\TYP{\Delta,x_B:\sharp A}{\vv{s}}{C}$, a substitution $\sigma\in\sem{\Delta}$ and value distributions $\vv{u_1},\vv{u_2}\in\sem{\sharp A}$, there are value distributions $\vv{w_1}, \vv{w_2}\in\sem{C}$ such that:
\[
\begin{array}{c}
    \vv{s}\ansubst{\sigma}{}\ansubst{\vv{u_1}/x}{B_1}{\ansubst{\vv{v_1}/y}{B_2}}\eval\vv{w_1}\\
    \vv{s}\ansubst{\sigma}{}\ansubst{\vv{u_2}/x}{B_1}{\ansubst{\vv{v_2}/y}{B_2}}\eval\vv{w_2}\\
\end{array}
\]

And, $\scal{\vv{w_1}}{\vv{w_2}} = \scal{\vv{u_1}}{\vv{u_2}}$.
\end{lemma}

\begin{proof}
    From the validity of the judgement of the form $\TYP{\Delta, x_A:\sharp A}{\vv{s}}{C}$, a substitution $\sigma\in\sem{\Delta}$, and value distributions $\vv{w_1},\vv{w_2}\in\sem{C}$ such that $\vv{s}\ansubst{\sigma}{}\ansubst{\vv{u_1}/x}{A}\eval\vv{w_1}$ and $\vv{s}\ansubst{\sigma}{}\ansubst{\vv{u_2}/x}{A}\eval\vv{w_2}$. In particular, we have that $\|\vv{w_1}\| = \|\vv{w_2}\|=1$. Applying \Cref{lem:VecRewrite}~four times, we know there are vectors $\vv{u_{01}},\vv{u_{02}},\vv{u_{03}},\vv{u_{04}}\in\sem{\sharp A}$ and scalars $\lambda_1,\lambda_2,\lambda_3,\lambda_4$ such that:
    
    \begin{align*}
        \vv{u_1} + \vv{u_2} = \lambda_1 \vv{u_{01}} & \vv{u_1} + i \vv{u_2} = \lambda_3 \vv{u_{03}} \\
        \vv{u_1} + (-1) \vv{u_2} = \lambda_2 \vv{u_{02}} & \vv{u_1} + (-i) \vv{u_2} = \lambda_4 \vv{u_{04}} \\
    \end{align*}

    From the validity of the judgement  $\TYP{\Delta, x_A:\sharp A}{\vv{s}}{C}$, we also know that there are value distributions $\vv{w_{01}},\vv{w_{02}},\vv{w_{03}},\vv{w_{04}}\in\sem{C}$ such that $\vv{s}\ansubst{\sigma}{}\ansubst{\vv{u_{0j}}}{}\eval\vv{w_{oj}}$ for all $f\in\{1\dotsb 4\}$. Combining the linearity of evaluation on the basis $A$ with the uniqueness of normal forms we deduce from what precedes that:

    \begin{align*}
        \vv{w_1} + \vv{w_2} = \lambda_1 \vv{w_{01}} & \vv{w_1} + i \vv{w_2} = \lambda_3 \vv{w_{03}} \\
        \vv{w_1} + (-1) \vv{w_2} = \lambda_2 \vv{w_{02}} & \vv{w_1} + (-i) \vv{w_2} = \lambda_4 \vv{w_{04}} \\
    \end{align*}

    Using the polarization identity (\Cref{prop:Polarization}), we conclude that:

    \begin{align*}
        &\scal{\vv{w_1}}{\vv{w_2}}\\
        &= \frac{1}{4}(\|\vv{w_1}+\vv{w_2}\| - \|\vv{w_1} + (-1)\vv{w_2}\| - i \|\vv{v_1} + i \vv{v_2}\| + i \|\vv{v_1} + (-i) \vv{v_2}\|)\\
        &= \frac{1}{4}((\lambda_1)^2\|\vv{w_{01}}\| - (\lambda_2)^2\|\vv{w_{02}}\| - i (\lambda_)^2 \|\vv{w_{03}}\| + i (\lambda_)^2\|\vv{w_{04}}\|)\\
        &= \frac{1}{4}((\lambda_1)^2\|\vv{u_{01}}\| - (\lambda_2)^2\|\vv{u_{02}}\| - i (\lambda_)^2 \|\vv{u_{03}}\| + i (\lambda_)^2\|\vv{u_{04}}\|)\\
        &= \frac{1}{4}(\|\vv{u_1}+\vv{u_2}\| - \|\vv{u_1} + (-1)\vv{u_2}\| - i \|\vv{u_1} + i \vv{u_2}\| + i \|\vv{u_1} + (-i) \vv{u_2}\|)\\
        &=\scal{u_1}{u_2}
    \end{align*}

\end{proof}

\begin{lemma}\label{lem:OrthogonalSubstitution} 
Given a valid typing judgement of the form $\TYP{\Delta, x_{B_1}:\sharp A_1, y_{B_2}: \sharp A_2}{\vv{s}}{C}$, a substitution $\sigma\in\sem{\Delta}$ and value distributions $\vv{u_1},\vv{u_2}\in\sem{\sharp A}$, there are value distributions $\vv{w_1},\vv{w_2}\in\sem{C}$ such that:
\[
\begin{array}{c}
    \vv{s}\ansubst{\sigma}{}\ansubst{\vv{u_1}/x}{B_1}{\ansubst{\vv{v_1}/y}{B_2}}\eval\vv{w_1}\\
    \vv{s}\ansubst{\sigma}{}\ansubst{\vv{u_2}/x}{B_1}{\ansubst{\vv{v_2}/y}{B_2}}\eval\vv{w_2}\\
\end{array}
\]
And, $\scal{\vv{w_1}}{\vv{w_2}} = 0$.
\end{lemma}

\begin{proof}
    From \Cref{lem:VecRewrite} we know that there are $\vv{u_0}\in\sem{\sharp A}, \vv{v_0}\in\sem{\sharp B}$ and $\lambda,\mu\in\C$ such that:
    \[
    \vv{u_2} + (-1) \vv{u_1} = \lambda\vv{u_0}\quad\text{and}\quad\vv{v_2} + (-1) \vv{v_1} = \mu \vv{v_0}
    \]
    For all $j,k\in\{0,1,2\}$, we have $\vv{s}\ansubst{\sigma}{}\ansubst{\vv{u_j}/x}{B_1}\ansubst{\vv{v_k}/y}{B_2}\eval\vv{w_{jk}}$. In particular, we can take $\vv{w_1}=\vv{w_{11}}$ and $\vv{w_2}=\vv{w_{22}}$. Now we observe that:
    \begin{enumerate}
        \item\label{A8:it1} $\vv{u_1}+\lambda\vv{u_0}= \vv{u_1} + \vv{u_2} + (-1) \vv{u_1}= \vv{u_2} + 0\vv{u_1}$, so that from linearity of substitution, linearity of evaluation and uniqueness of normal forms, we get:
        \[
        \begin{array}{c c}
            \begin{array}{c}
                \vv{w_{1k}} + \lambda\vv{w_{0k}} = \vv{w_{2k}} + 0 \vv{w_{1k}}\\
                \vv{w_{2k}} + (-\lambda)\vv{w_{0k}} = \vv{w_{1k}} + 0 \vv{w_{2k}}
            \end{array}&
            (\text{for all }k\in\{0,1,2\})
        \end{array}
        \]
        
        \item\label{A8:it2} $\vv{v_1}+\mu\vv{v_0}= \vv{v_1} + \vv{v_2} + (-1) \vv{v_1}= \vv{v_2} + 0\vv{v_1}$, so that from linearity of substitution, linearity of evaluation and uniqueness of normal forms, we get:
        \[
        \begin{array}{c c}
            \begin{array}{c}
                \vv{w_{j1}} + \mu\vv{w_{j0}} = \vv{w_{j2}} + 0 \vv{w_{j1}}\\
                \vv{w_{j2}} + (-\mu)\vv{w_{j0}} = \vv{w_{j1}} + 0 \vv{w_{j2}}
            \end{array}&
            (\text{for all }j\in\{0,1,2\})
        \end{array}
        \]
        
        \item\label{A8:it3} $\scal{\vv{u_1}}{\vv{u_2}}=0$, so that from \Cref{lem:InnerProdSingleVar}~we get $\scal{\vv{w_{1k}}}{\vv{w_{2k}}}=0$ (for all $k\in\{0,1,2\}$).
        
        \item\label{A8:it4} $\scal{\vv{v_1}}{\vv{v_2}}=0$, so that from \Cref{lem:InnerProdSingleVar}~we get $\scal{\vv{w_{j1}}}{\vv{w_{j2}}}=0$ (for all $j\in\{0,1,2\}$).
    \end{enumerate}

    From the above, we get:
    \begin{align*}
        \scal{\vv{w_1}}{\vv{w_2}} &= \scal{\vv{w_{11}}}{\vv{w_{22}}} = \scal{\vv{w_{11}}}{\vv{w_{22}}+0\vv{w_{12}}} & \\
        &=\scal{\vv{w_{11}}}{\vv{w_{12}}+ \lambda\vv{w_{02}}} & (\text{from \Cref{A8:it1}, } k=2)\\
        &=\scal{\vv{w_{11}}}{\vv{w_{12}}} + \lambda \scal{\vv{w_{11}}}{\vv{w_{02}}} &\\
        &= 0 + \lambda \scal{\vv{w_{11}}}{\vv{w_{02}}} & (\text{from \Cref{A8:it4}, } j=1)\\
        &= \lambda \scal{\vv{w_{11}} + 0\vv{w_{21}}}{\vv{w_{02}}} & \\
        &= \lambda \scal{\vv{w_{21}} + (-\lambda)\vv{w_{01}}}{\vv{w_{02}}} & (\text{from \Cref{A8:it1}, } k=1)\\
        &= \lambda \scal{\vv{w_{21}}}{\vv{w_{02}}} - |\lambda|^2 \scal{\vv{w_{01}}}{\vv{w_{02}}} & \\
        &= \lambda \scal{\vv{w_{21}}}{\vv{w_{02}}} - 0 & (\text{from \Cref{A8:it4}, } j=0)\\
        &=\scal{\vv{w_{21}}}{\vv{w_{22}}- \vv{w_{12}}} & \\
        &=\scal{\vv{w_{21}}}{\vv{22}} - \scal{\vv{w_{21}}}{\vv{w_12}} & \\
        &= 0 - \scal{\vv{w_{21}}}{\vv{w_{12}}} & (\text{from \Cref{A8:it4}, } j=2)\\
    \end{align*}
    Hence $\scal{\vv{w_1}}{\vv{w_2}} = \scal{\vv{w_{11}}}{\vv{w_{22}}} = - \scal{\vv{w_{21}}}{\vv{w_{12}}}$. Exchanging the indices in the previous reasoning, we also get 
    \[
    \scal{\vv{w_1}}{\vv{w_2}}=-\scal{\vv{w_{21}}}{\vv{w_{12}}}=-\scal{\vv{w_{12}}}{\vv{w_{21}}}
    \]
    So that we have:
    \[
        \scal{\vv{w_1}}{\vv{w_2}}=-\scal{\vv{w_{21}}}{\vv{w_{12}}}=-\overline{\scal{\vv{w_{21}}}{\vv{w_{12}}}}\in\R
    \]
    If we now replace $\vv{u_2}\in\sem{\sharp A}$ with $i\vv{u_2}\in\sem{\sharp A}$, the very same technique allows us to prove that $i\scal{\vv{w_1}}{\vv{w_2}}=\scal{\vv{w_1}}{i \vv{w_2}}\in\R$. Therefore, $\scal{\vv{w_1}}{\vv{w_2}}=0$.
\end{proof}

\begin{lemma}\label{lem:UnitPreserTens} 
Given a valid typing judgement of the form $\TYP{\Delta,x_{B_1}:\sharp A_1, y_{B_2}:\sharp A_2}{\vv{s}}{C}$, a substitution $\sigma\in\sem{\Delta}$, and value distributions $\vv{u_1},\vv{u_2}\in\sem{\sharp A}$ and $\vv{v_1},\vv{v_2}\in\sem{\sharp B}$, there are value distributions $\vv{w_1},\vv{w_2}\in\sem{C}$ such that:
\[
\begin{array}{c}
    \vv{s}\ansubst{\sigma}{}\ansubst{\vv{u_1}/x}{B_1}{\ansubst{\vv{v_1}/y}{B_2}}\eval\vv{w_1}\\
    \vv{s}\ansubst{\sigma}{}\ansubst{\vv{u_2}/x}{B_1}{\ansubst{\vv{v_2}/y}{B_2}}\eval\vv{w_2}\\
\end{array}
\]

And, $\scal{\vv{w_1}}{\vv{w_2}} = \scal{\vv{u_1}}{\vv{u_2}} \scal{\vv{v_1}}{\vv{v_2}}$.

\begin{proof}
    Let $\alpha=\scal{\vv{u_1}}{\vv{u_2}}$ and $\beta=\scal{\vv{v_1}}{\vv{v_2}}$. We observe that:
    \[
    \scal{\vv{u_1}}{\vv{u_2}+(-\alpha)\vv{u_1}} = \scal{\vv{u_1}}{\vv{u_2}} - \alpha \scal{\vv{u_1}}{\vv{u_1}} = \alpha - \alpha = 0
    \]
    And similarly that, $\scal{\vv{v_1}}{\vv{v_2}+ (-\beta) \vv{v_1}} = 0$. From \Cref{lem:VecRewrite}, we know that there are $\vv{u_0}\in\sem{\sharp A}$, $\vv{v_0}\in\sem{\sharp B}$ and $\lambda,\mu\in\C$ such that:
    \begin{align*}
        \vv{u_2} +(-\alpha)\vv{u_1} = \lambda\vv{u_0}& \text{ and } & \vv{v_2} + (-\beta)\vv{v_1} = \mu\vv{v_0} 
    \end{align*}
    For all $j,k\in\{0,1,2\}$, we have$\ansubst{\sigma}{}\ansubst{\vv{u_j}/x}{B_1}\ansubst{\vv{v_k}/y}{B_2}\in\sem{\Delta,x_{B_1}:\sharp A_1, y_{B_2}:\sharp A_2}$, hence there is $\vv{w_{jk}}\in\sem{C}$ such that:
    \[
    \vv{s}\ansubst{\sigma}{}\ansubst{u_j/x}{B_1}\ansubst{\vv{v_k}/y}{B_2}\eval\vv{w_{jk}}
    \]
    In particular, we can take $\vv{w_1}=\vv{w_{11}}$ and $\vv{w_2}=\vv{w_{22}}$. Now we observe that:
    \begin{enumerate}
        \item\label{A9:it1} $\lambda \vv{u_0} + \alpha\vv{u_1}=\vv{u_2} + (-\alpha) \vv{u_1} + \alpha \vv{u_1} = \vv{u_2} + 0 \vv{u_1}$, so that from the linearity of the substitution, linearity of evaluation and uniqueness of normal forms, we get:
        \[
        \lambda\vv{w_{0k}} + \alpha \vv{w_{1k}} = \vv{w_{2k}} + 0 \vv{w_{1k}} \qquad(\text{for all }k\in\{0,1,2\})
        \]
        
        \item\label{A9:it2} $\mu\vv{v_0} + \beta\vv{v_1}=\vv{v_2} + (-\beta) \vv{v_1} + \beta \vv{v_1} = \vv{v_2} + 0 \vv{v_1}$, so that from the linearity of the substitution, linearity of evaluation and uniqueness of normal forms, we get:
        \[
        \mu\vv{w_{j0}} + \beta \vv{w_{j1}} = \vv{w_{j2}} + 0 \vv{w_{j1}} \qquad(\text{for all }j\in\{0,1,2\})
        \]
        
        \item\label{A9:it3} $\scal{\vv{u_1}}{\lambda\vv{u_0}}=\scal{\vv{u_1}}{\vv{u_2}+(- \alpha)\vv{u_1}}=0$, so that from \Cref{lem:InnerProdSingleVar} we get:
        \[
        \scal{\vv{w_{1k}}}{\lambda\vv{w_{0k}}}= 0 \qquad(\text{for all }k\in\{0,1,2\})
        \]
        
        \item\label{A9:it4} $\scal{\vv{v_1}}{\mu\vv{v_0}}=\scal{\vv{v_1}}{\vv{v_2} + (-\beta)}\vv{v_1}=0$, so that from \Cref{lem:InnerProdSingleVar} we get:
        \[
        \scal{\vv{w_{j1}}}{\mu\vv{w_{j0}}}=0
        \]
                
        \item\label{A9:it5} $\scal{\vv{u_1}}{\lambda\vv{u_0}}=\scal{\vv{v_1}}{\mu\vv{v_0}}=0$ so that from \Cref{lem:OrthogonalSubstitution} we get:
        \[
        \scal{\vv{w_{11}}}{\lambda\mu\vv{w_{00}}}=0
        \]
        (Again the equality $\scal{\vv{w_{11}}}{\lambda\mu\vv{w_{00}}}$ is trivial when $\lambda=0$ or $\mu=0$. When $\lambda ,\mu\neq 0$ we deduce from the above that $\scal{\vv{u_1}}{\vv{u_0}}=\scal{\vv{v_1}}{\vv{v_0}}=0$, from which we get $\scal{\vv{w_{11}}}{\vv{w_{00}}}=0$ by \Cref{lem:OrthogonalSubstitution})
    \end{enumerate}

    From above, we get:
    \begin{align*}
        \vv{w_{22}} + 0\vv{w_{12}} + &0\vv{w_{01}} + 0\vv{w_{11}} \\
        &= \lambda\vv{w_{02}} + \alpha\vv{w_{12}} + 0\vv{w_{01}} + 0\vv{w_{11}} & (\text{from~\Cref{A9:it1}}, k =1)\\
        &= \lambda(\vv{w_{02}}+0\vv{w_{01}}) + \alpha(\vv{w_{12}}+0\vv{w_{11}})&\\
        &= \lambda(\mu\vv{w_{00}} + \beta\vv{w_{01}}) + \alpha(\mu\vv{w_{01}}+\beta\vv{w_{11}}) & (\text{from~\Cref{A9:it2}}, j=0,1)\\
        &= \lambda\mu\vv{w_{00}} + \lambda\beta\vv{w_{01}} + \alpha\mu\vv{w_{10}} + \alpha\beta\vv{w_{11}}
    \end{align*}
    Therefore:
    \begin{align*}
        &\scal{\vv{w_1}}{\vv{w_2}} \\
        &= \scal{\vv{w_{11}}}{\vv{w_{22}} + 0 \vv{w_{12}} + 0 \vv{w_{01}} + 0 \vv{w_{11}}}\\
        &= \scal{\vv{w_{11}}}{\lambda\mu\vv{w_{00}} + \lambda\beta\vv{w_{01}} + \alpha\mu\vv{w_{10}} + \alpha\beta\vv{w_{11}}}\\
        &=\scal{\vv{w_{11}}}{\lambda\mu\vv{w_{00}}} + \scal{\vv{w_{11}}}{\lambda\beta\vv{w_{01}}} + \scal{\vv{w_{11}}}{\alpha\mu\vv{w_{10}}} + \scal{\vv{w_{11}}}{\alpha\beta\vv{w_{11}}}\\
        &=\lambda\mu\scal{\vv{w_{11}}}{\vv{w_{00}}} + \lambda\beta\scal{\vv{w_{11}}}{\vv{w_{01}}} + \alpha\mu\scal{\vv{w_{11}}}{\vv{w_{10}}} + \alpha\beta\scal{\vv{w_{11}}}{\vv{w_{11}}}\\
        &= 0 + 0 + 0 + \alpha\beta = \scal{\vv{u_1}}{\vv{u_2}}\scal{\vv{v_1}}{\vv{v_2}}
    \end{align*}
    From~\Cref{A9:it5,A9:it3,A9:it4} with $j=1$ and concluding with the definition of $\alpha$ and $\beta$.
\end{proof}
\end{lemma}

Now, we can restate and prove \Cref{thm:TypingRulesValidity}.
\begin{restatetheorem}[Restatement of \Cref{thm:TypingRulesValidity}]
  \itshape
  All the typing rules in \Cref{tab:TypingRules} are valid.
\end{restatetheorem}
\begin{proof}
    For each typing rule in \Cref{tab:TypingRules}~we have to show the typing judgement is valid starting from the premises:
    \begin{description}
    \item[Axiom] It is clear that $\sdom{x:A}\subseteq\{x\}=\dom{x:A}$. Moreover, given $\sigma\in\sem{x^B:A}$, we have $\sigma=\ansubst{\vv v/x}{B}$ for some $\vv{v}\in\sem{A}$. Therefore, $x\ansubst{\sigma}{}=x\ansubst{\vv v}{B}=\vv{v}\real A$.
    
    \item[UnitLam] If the hypothesis is valid, $\sdom{\Gamma,x^X:A}\subseteq \FV{\sum_{i=1}^{n}\alpha_i \vv{t_i}}\subseteq \dom{\Gamma,x^X:A}$. It follows that $\sdom{\Gamma}\subseteq \FV{\sum_{i=1}^{n}\alpha_i (\Lam{x}{X}{\vv{t_i}})}\subseteq \dom{\Gamma}$. Given $\sigma\in\sem{\Gamma}$, we want to show that $(\sum_{i=1}^{n}\alpha_i (\Lam{x}{X}{\vv{t_i}}))\ansubst{\sigma}{}\real A\Arr B$. Let $\vv v\in\sem{A}$, then:
    
    \begin{align*}
        (\sum_{i=1}^{n} \alpha_i(\Lam{x}{X}{\vv{t_i}}))\ansubst{\sigma}{} \vv v&= (\sum_{j=1}^{m} \beta_j (\sum_{i=1}^{n} \alpha_i (\Lam{x}{X}{\vv{t_i}}) [\sigma_i])) \vv{v} \\
        &= (\sum_{i=1}^{n}\sum_{j=1}^{m}\alpha_i\beta_j (\Lam{x}{X}{\vv{t_i}[\sigma_j]}))\vv v\\
        &\lra \sum_{i=1}^{n}\sum_{j=1}^{m}\alpha_i\beta_j \vv{t_i}[\sigma_j]\ansubst{\vv v/x}{X}\\
        &=\sum_{i=1}^{n}\alpha_i \vv{t_i}\ansubst{\sigma}{}\ansubst{\vv v/x}{X}\\
        &=(\sum_{i=1}^{n}\alpha_i \vv{t_i})\ansubst{\sigma}{}\ansubst{\vv v/x}{X}\qquad{\text{By \Cref{lem:distributiveSubstitution}}}
    \end{align*}
    
    Considering that $\ansubst{\sigma}{}\in\sem{\Gamma}$, then we have that $\ansubst{\sigma}{}\ansubst{\vv v/x}{X}\in\sem{\Gamma,x^X:A}$. Since we assume $\TYP{\Gamma, x^X:A}{\sum_{i=1}^{n}\alpha_i\vv{t_i}}{B}$, then $\vv{t_i}\ansubst{\sigma}{}\ansubst{\vv v/x}{X}\real B$. Finally, we can conclude that the distribution: $\sum_{i=1}^{n}\alpha_i (\Lam{x}{X}{\vv{t_i}})\in\sem{A\Arr B}$.

    \item[App] If the hypotheses are valid, then:
    \begin{itemize}
        \item $\sdom{\Gamma}\subseteq \FV{\vv s}\subseteq \dom{\Gamma}$ and $\vv s \ansubst{\sigma_\Gamma}{}\Vdash A\Arr B\ \forall \sigma_\Gamma\in\sem{\Gamma}$.
        \item $\sdom{\Delta}\subseteq \FV{\vv t}\subseteq \dom{\Delta}$ and $\vv t\ansubst{\sigma_\Delta}{}\Vdash A\ \forall\sigma_\Delta\in\sem{\Delta}$.
    \end{itemize}
    
    From this, we can conclude that $\sdom{\Gamma,\Delta}\subseteq \FV{\vv s \vv t}\subseteq \dom{\Gamma,\Delta}$. Given $\sigma\in\sem{\Gamma,\Delta}$, we can observe that $\sigma=\sigma_\Gamma,\sigma_\Delta$ for some $\sigma_\Gamma\in\sem{\Gamma}$ and $\sigma_\Delta\in\sem{\Delta}$. Then we have:
    
    \begin{align*}
        (\vv{t}\vv{s})\ansubst{\sigma}{} &= (\vv{t}\vv{s})\ansubst{\sigma_\Gamma}{}\ansubst{\sigma_\Delta}{}\\
        &=(\sum_{i=i}^{n}\alpha_i (\vv{t}\vv{s})[\sigma_{\Gamma i}])\ansubst{\sigma_\Delta}{}\\
        &=\sum_{j=1}^{m} \beta_j (\sum_{i=1}^{n} \alpha_i (\vv{t} \vv{s})[\sigma_{\Gamma i}])[\sigma_{\Delta j}]\\
        &=\sum_{i=1}^{n}\sum_{j=1}^{m} \alpha_i \beta_j \vv{t}\,[\sigma_{\Gamma i}][\sigma_{\Delta j}] \vv{s}\,[\sigma_{\Gamma i}][\sigma_{\Delta j}]\\
        &=\sum_{i=1}^{n}\sum_{j=1}^{m} \alpha_i \beta_j \vv{t}\,[\sigma_{\Gamma i}]\vv{s}\,[\sigma_{\Delta j}]\\
        &\equiv (\sum_{i=1}^{n}\alpha_i\vv{t}[\sigma_{\Gamma i}])(\sum_{j=1}^{m} \beta_j \vv{s}[\sigma_{\Delta j}])\\
        &=\vv{t}\ansubst{\sigma_\Gamma}{} \vv{s}\ansubst{\sigma_\Delta}{}\\
        &\eval (e^{i\theta_{1}} \vv{v}) (e^{i\theta_{2}} \vv{w})\qquad\text{Where: } \vv{v}\in\sem{A\Arr B}, \vv{w}\in\sem{A}\\
        &\equiv e^{i\theta} (\vv{v} \vv{w})\qquad\text{with: }\theta=\theta_1 + \theta_2\\
        &\lraneq e^{i\theta}\vv r\qquad\text{where: } \vv{r}\real B
    \end{align*}
    
    Then we can conclude that $(\vv{t}\vv{s})\ansubst{\sigma}{}\real B$.

    \item[Pair] If the hypotheses are valid, then:

    \begin{itemize}
        \item $\sdom{\Gamma}\subseteq \FV{\vv s}\subseteq \dom{\Gamma}$ and $\vv s \ansubst{\sigma_\Gamma}{}\Vdash A\ \forall \sigma_\Gamma\in\sem{\Gamma}$.
        \item $\sdom{\Delta}\subseteq \FV{\vv t}\subseteq \dom{\Delta}$ and $\vv t\ansubst{\sigma_\Delta}{}\Vdash B\ \forall \sigma_\Delta\in\sem{\Delta}$.
    \end{itemize}
    
    From this, we can conclude that $\sdom{\Gamma,\Delta}\subseteq \FV{(\vv s, \vv t)}\subseteq \dom{\Gamma,\Delta}$. Given $\sigma\in\sem{\Gamma,\Delta}$, we can observe that $\sigma=\sigma_\Gamma,\sigma_\Delta$ for some  $\sigma_\Gamma\in\sem{\Gamma}$ and $\sigma_\Delta\in\sem{\Delta}$. Then we have:

    \begin{align*}
        \Pair{\vv{t}}{\vv{s}}\ansubst{\sigma}{} &= \Pair{\vv{t}}{\vv{s}}\ansubst{\sigma_\Gamma}{}\ansubst{\sigma_\Delta}{}\\
        &=\sum_{j=1}^{m} \beta_j (\sum_{i=1}^{n} \alpha_i \Pair{\vv{t}}{\vv{s}}[\sigma_{\Gamma i}])[\sigma_{\Delta j}]\\
        &\equiv\sum_{i=1}^{n}\sum_{j=1}^{m} \alpha_i \beta_j \Pair{\vv{t}\,[\sigma_{\Gamma i}][\sigma_{\Delta j}]}{\vv{s}\,[\sigma_{\Gamma i}][\sigma_{\Delta j}]}\\
        &=\sum_{i=1}^{n}\sum_{j=1}^{m} \alpha_i \beta_j \Pair{\vv{t}\,[\sigma_{\Gamma i}]}{\vv{s}\,[\sigma_{\Delta j}]}\\
        &=\Pair{\sum_{i=1}^{n} \alpha_i \vv{t}\, [\sigma_{\Gamma i}]}{\sum_{j=1}^{m} \beta_j \vv{s}\, [\sigma_{\Delta j}]}\\
        &=\Pair{\vv{t}\ansubst{\sigma_\Gamma}{}}{\vv{s}\ansubst{\sigma_\Delta}{}}\\
        &\eval \Pair{e^{i\theta_1}\vv v}{e^{i\theta_2}\vv w}\qquad\text{where: }\vv{v}\in\sem{A}, \vv{w}\in\sem{B}\\
        &= e^{i\theta} \Pair{\vv{v}}{\vv{w}}\qquad\text{where: }\theta=\theta_1 + \theta_2
    \end{align*}
    
    From this we can conclude that $\Pair{\vv t}{\vv{s}}\ansubst{\sigma}{}\real A\times B$.
    
    \item[LetPair] If the hypotheses are valid, then:
    \begin{itemize}
        \item $\sdom{\Gamma}\subseteq \FV{\vv t} \subseteq \dom{\Gamma}$ and $\vv t \ansubst{\sigma_\Gamma}{}\Vdash A\times B\ \forall \sigma_\Gamma\in\sem\Gamma$
        \item $\sdom{\Delta, {x}^{X}:A, {y}^{Y}:B}\subseteq\FV{\vv s}$
        \item $\FV{\vv{s}}\subseteq \dom{\Delta, {x}^{X}:A, {y}^{Y}:B}$
        \item $\vv s \ansubst{\sigma_\Delta}{}\Vdash C\ \forall \sigma_\Delta\in\sem{\Delta, {x}^{X}:A, {y}^{Y}:B}$
    \end{itemize}
    From this, we can conclude that:
    \begin{itemize}
        \item $\sdom{\Gamma,\Delta}\subseteq\FV{\LetP{x}{X}{y}{Y}{\vv{s}}{\vv{t}}}$
        \item $\FV{\LetP{x}{X}{y}{Y}{\vv{s}}{\vv{t}}}\subseteq\dom{\Gamma,\Delta}$
    \end{itemize}
    
    Given $\sigma\in\sem{\Gamma,\Delta}$, we have that $\ansubst{\sigma}{}=\ansubst{\sigma_\Gamma}{},\ansubst{\sigma_\Delta}{}$ for some $\sigma_\Gamma\in\sem\Gamma$ and $\sigma_\Delta\in\sem\Delta$. Then we have:
    \begin{align*}
        (&\LetP{x}{X}{y}{Y}{\vv{t}}{\vv{s}})\ansubst{\sigma}{} = \\
        &(\LetP{x}{X}{y}{Y}{\vv{t}}{\vv{s}})\ansubst{\sigma_\Gamma}{}\ansubst{\sigma_\Delta}{}\\
        &= (\sum_{i=1}^{n}\alpha_i(\LetP{x}{X}{y}{Y}{\vv{t}}{\vv{s}})[\sigma_{\Gamma i}])\ansubst{\sigma_{\Delta j}}{}\\
        &\equiv (\LetP{x}{X}{y}{Y}{\sum_{i=1}^{n}\alpha_i[\sigma_{\Gamma i}]\vv{t}}{\vv{s}})\ansubst{\sigma_\Delta}{}\\
        &= (\LetP{x}{X}{y}{Y}{\vv{t}\ansubst{\sigma_\Gamma}{}}{\vv{s}})\ansubst{\sigma_\Delta}{}\\
        &\eval (\LetP{x}{X}{y}{Y}{e^{i\theta} \Pair{\vv{v}}{\vv{w}}}{\vv{s}})\ansubst{\sigma_\Delta}{}\\
        &\hspace*{4cm}{\text{Where: }}\vv{v}\in\sem{A},\vv{w}\in\sem{B}\\
        &\lra e^{i\theta_1} (\vv{s}\ansubst{\sigma_\Delta}{}\ansubst{\Pair{\vv{v}}{\vv{w}}/x\otimes y}{X\otimes Y})\\
        &= e^{i\theta_1} (\vv{s}\ansubst{\sigma_\Delta}{}\ansubst{\vv{v}/x}{X}\ansubst{\vv{w}/y}{Y})\\
        &\eval e^{i\theta_1}  (e^{i\theta_2}  \vv{u})\qquad\text{where: }\vv{u}\in\sem{C}\\
        &\equiv e^{i\theta}  \vv{u}\qquad\text{where: }\theta=\theta_1 + \theta_2
    \end{align*}
    
    Since $\ansubst{\sigma_\Delta}{}\ansubst{\vv{v}/x}{X}\ansubst{\vv{w}/y}{Y}\in\sem{\Delta,{x}^{X}:A,{y}^Y:B}$, then we can conclude that $(\LetP{x}{X}{y}{Y}{\vv{t}}{\vv{s}})\ansubst{\sigma}{}\real C$.

    \item[LetTens] If the hypotheses are valid then:
    \begin{itemize}
        \item $\sdom{\Gamma}\subseteq \FV{\vv t} \subseteq \dom{\Gamma}$ and $\vv t \ansubst{\sigma}{}\Vdash \sharp(A\times B)\ \forall \sigma\in\sem\Gamma$
        \item $\sdom{\Delta, {x}^{X}:\sharp A, {y}^{Y}:\sharp B}\subseteq \FV{\vv s}$
        \item $\subseteq \dom{\Delta, {x}^{X}:\sharp A, {y}^{Y}:\sharp B}$
        \item $\vv s \ansubst{\sigma}{}\Vdash \sharp C\ \forall \sigma\in\sem{\Delta, {x}^{X}:\sharp A, {y}^{Y}:\sharp B}$
    \end{itemize}
    
    From this we can conclude that:
    \begin{itemize}
        \item $\sdom{\Gamma,\Delta}\subseteq\FV{\LetP{x}{X}{y}{Y|}{\vv{t}}{\vv{s}}}$
        \item $\FV{\LetP{x}{X}{y}{Y}{\vv{t}}{\vv{s}}}\subseteq\dom{\Gamma,\Delta}$
    \end{itemize}
    
    Given $\sigma\in\sem{\Gamma,\Delta}$, we have that $\ansubst{\sigma}{}=\ansubst{\sigma_\Gamma}{},\ansubst{\sigma_\Delta}{}$ for some $\sigma_\Gamma\in\sem\Gamma$ and $\sigma_\Delta\in\sem\Delta$. Using the first hypothesis we have that, $\vv t\ansubst{\sigma_\Gamma}{}\real \sharp(A\times B)$, from \Cref{thm:SharpCharacterization} we have that:
    
    \[\vv t\ansubst{\sigma_\Gamma}{}\eval e^{i\theta_1}\vv{u}=e^{i\theta_1}(\sum_{k=1}^{l} \gamma_k \Pair{\vv{v_k}}{\vv{u_k}})\] 
    
    With:
    \begin{itemize}
        \item $\sum_{k=1}^{l} |\gamma_k|^2 = 1$
        \item $\forall k,\ \vv{v_k}\in\sem{A},\ \vv{u_k}\in\sem{B}$
        \item $\forall k\neq l, \scal{\Pair{\vv{v_k}}{\vv{u_k}}}{\Pair{\vv{v_l}}{\vv{u_l}}}= 0$
    \end{itemize}
    
    Then:
    \begin{align*}
        (&\LetP{x}{X}{y}{Y}{\vv{t}}{\vv{s}})\ansubst{\sigma}{} \\
        &= \LetP{x}{X}{y}{Y}{\vv{t}}{\vv{s}}\ansubst{\sigma_\Gamma}{}\ansubst{\sigma_\Delta}{}\\
        &=(\sum_{i=1}^{n}\alpha_i\LetP{x}{X}{y}{Y}{\vv{t}}{\vv{s}}\ [\sigma_{\Gamma i}])\ansubst{\sigma_{\Delta}}{}\\
        &\equiv (\LetP{x}{X}{y}{Y}{\sum_{i=1}^{n}\alpha_i\vv{t}\ [\sigma_{\Gamma i}]}{\vv{s}})\ansubst{\sigma_\Delta}{}\\
        &=(\LetP{x}{X}{y}{Y}{\vv{t}\ansubst{\sigma_\Gamma}{}}{\vv{s}})\ansubst{\sigma_\Delta}{}\\
        &\eval(\LetP{x}{X}{y}{Y}{e^{i\theta_1}\vv{u}}{\vv{s}})\ansubst{\sigma_\Delta}{}\\
        &\lra e^{i\theta_1}(\vv{s}\ansubst{\sigma_\Delta}{}\ansubst{\vv{u}/x\otimes y}{X\otimes Y})\\
        &=e^{i\theta_1} (\sum_{k=1}^{l}\gamma_k\vv{s}\ansubst{\sigma_\Delta}{}\ansubst{\vv{v_k}/x}{X}\ansubst{\vv{u_k}/y}{Y})\\
        &\eval e^{i\theta_1} (\sum_{k=1}^{l}\gamma_k e^{i\rho_k} \vv{w_k})\quad\text{where: }\vv{w_k}\in\sem{C}\\
    \end{align*}

    It remains to be seen that the term has norm-$1$, $\|\sum_{k=1}^{l}\gamma_k e^{i\rho_k} \vv{w_k}\|=1$. For that, we observe:
    \begin{align*}
        \|&\sum_{k=1}^{l}\gamma_k e^{i\rho_k} \vv{w_k}\| \\
        &= \scal{\sum_{k=1}^{l}\alpha_i e^{i\rho_k} \vv{w_k}}{\sum_{k'=1}^{l}\gamma_{k'} e^{i\rho_{k'}} \vv{w_{k'}}}\\
        &= \sum_{k=1}^{l}\sum_{k'}^{l}\overline{\gamma_k e^{i\rho_k}}\  \gamma_{k'} e^{i\rho_{k'}}\scal{\vv{w_k}}{\vv{w_{k'}}}\\
        &=\sum_{k=1}^{l}\sum_{k'=1}^{l}\overline{\gamma_k e^{i\rho_k}}\ \gamma_{k'} e^{i\rho_{k'}} \scal{\vv{v_k}}{\vv{v_{k'}}}\scal{\vv{u_k}}{\vv{u_{k'}}}\quad(\text{from \Cref{lem:UnitPreserTens}})\\
        &= \sum_{k=1}^{k}\sum_{k'=1}^{l}\overline{\gamma_k e^{i\rho_k}}\  \gamma_{k'} e^{i\rho_{k'}} \scal{\Pair{\vv{u_k}}{\vv{v_k}}}{\Pair{\vv{u_{k'}}}{\vv{v_{k'}}}}\quad(\text{from \Cref{prop:InnerProdPairs}})\\
        &=\sum_{k=1}^n \overline{\gamma_k e^{i\rho_k}}\ \gamma_k e^{i\rho_k} \scal{\Pair{\vv{v_k}}{\vv{u_k}}}{\Pair{\vv{v_k}}{\vv{u_k}}} \\
        & \quad + \sum_{k,k'=1; k\neq k'}^n \overline{\gamma_k e^{i\rho_k}}\  \gamma_{k'} e^{i\rho_{k'}} \scal{\Pair{\vv{v_k}}{\vv{u_k}}}{\Pair{\vv{v_{k'}}}{\vv{u_{k'}}}}\\
        &= \sum_{k=1}^n \overline{\gamma_k e^{i\rho_k}}\ \gamma_k e^{i\rho_k} + 0 \\
        &= \sum_{k=1}^{l} |\gamma_k|^2 |e^{i\rho_k}|^2 = 1
    \end{align*}

    Then $\sum_{k=1}^{l}\gamma_k e^{i\rho_k} \vv{w_k}\in\sem{\sharp C}$. Finally, we can conclude that: $(\LetP{x}{X}{y}{Y}{\vv{t}}{\vv{s}})\ansubst{\sigma}{}\real{\sharp C}$

    \item[Case] If the hypotheses are valid then:
    \begin{itemize}
        \item $\sdom{\Gamma}\subseteq \FV{\vv{t}}\subseteq \dom{\Gamma}$
        \item For every $\sigma_\Gamma\in\sem{\Gamma}$, $\vv{t}\ansubst{\sigma_\Gamma}{}\real\genbasis{\vv{v_i}}{i=1}{n}$
        \item For every $i\in\{0,\dotsb ,n\}, \sdom{\Delta}\subseteq \FV{\vv{s_i}}\subseteq \dom{\Delta}$
        \item For every $i\in\{0,\dotsb ,n\}, \sigma_\Delta\in\sem{\Delta}$, $\vv{s_i}\ansubst{\sigma_\Delta}{}\real A$
    \end{itemize}

    From this we can conclude that:

    \begin{itemize}
        \item $\sdom{\Gamma,\Delta}\subseteq \FV{\gencase{\vv{t}}{\vv{v_1}}{\vv {v_n}}{\vv{s_1}}{\vv{s_n}}}$
        \item $\FV{\gencase{\vv{t}}{\vv{v_1}}{\vv {v_n}}{\vv{s_1}}{\vv{s_n}}}\subseteq \dom{\Gamma,\Delta}$
    \end{itemize}
    
    Then, given $\sigma\in\sem{\Gamma,\Delta}$, we have that $\ansubst{\sigma}{}=\ansubst{\sigma_\Gamma}{}\ansubst{\sigma_\Delta}{}$ for some $\sigma_\Gamma\in\sem{\Gamma}$ and $\sigma_\Delta\in\sem{\Delta}$. Using the first hypothesis we have that, $\vv{t}\ansubst{\sigma_\Gamma}{}\eval e^{i\theta_1}\vv{v_k}$ for some $k\in\{1,\dotsb ,n\}$. From the second hypothesis we have that $\vv{s_i}\ansubst{\sigma_\Delta}{}\eval e^{i\rho_i}\vv{u_i}\in\sem{A}$ for $i\in\{1,\dotsb , n\}$. Therefore:

    \begin{align*}
        (&\gencase{\vv{t}}{\vv{v_1}}{\vv{v_n}}{\vv{s_1}}{\vv{s_n}})\ansubst{\sigma}{}\\ 
        &= (\gencase{\vv{t}}{\vv{v_1}}{\vv{v_n}}{\vv{s_1}}{\vv{s_n}})\ansubst{\sigma_\Gamma}{}\ansubst{\sigma_\Delta}{}\\
        &= (\sum_{i=1}^{n}\alpha_i \gencase{\vv{t}[\sigma_{\Gamma i}]}{\vv{v_1}}{\vv{v_n}}{\vv{s_1}}{\vv{s_n}})\ansubst{\sigma_\Delta}{} \\
        &\equiv (\gencase{\sum_{i=1}^{n} \alpha_i \vv{t}[\sigma_{\Gamma i}]}{\vv{v_1}}{\vv{v_n}}{\vv{s_1}}{\vv{s_n}})\ansubst{\sigma_\Delta}{}\\
        &=(\gencase{\vv{t}\ansubst{\sigma_\Gamma}{}}{\vv{v_1}}{\vv{v_n}}{\vv{s_1}}{\vv{s_n}})\ansubst{\sigma_\Delta}{}\\
        &\eval(\gencase{e^{i\theta_1}\vv{v_k}}{\vv{v_1}}{\vv{v_n}}{\vv{s_1}}{\vv{s_n}})\ansubst{\sigma_\Delta}{}\\
        &\lraneq e^{i\theta_1}(\vv{s_k}\ansubst{\sigma_\Delta}{})\\
        &\eval e^{i\theta_1}(e^{i\rho_k} \vv{u_k})\qquad\text{Where: }\vv{u_k}\in\sem{A}\\
        &\equiv e^{i\theta} \vv{u_k}\qquad\text{With: }\theta=\theta_1 +\theta_2
    \end{align*}
    
    Since we pose no restriction on $k$, we can conclude that: $(\gencase{\vv{t}}{\vv{v_1}}{\vv{v_n}}{\vv{s_1}}{\vv{s_n}})\ansubst{\sigma}{}\real A$

    \item[UnitCase] If the hypotheses are valid, then:
    \begin{itemize}
        \item $\sdom{\Gamma}\subseteq \FV{\vv{t}}\subseteq \dom{\Gamma}$
        \item For every $\sigma_\Gamma\in\sem{\Gamma}$, $\vv{t}\ansubst{\sigma_\Gamma}{}\real\sharp\genbasis{\vv{v_i}}{i=1}{n}$
        \item For every $i\in\{0,\dotsb ,n\}$, $\sdom{\Delta}\subseteq \FV{\vv{s_i}}\subseteq \dom{\Delta}$
        \item For every $i\in\{0,\dotsb ,n\}, \sigma_\Delta\in\sem{\Delta}$, $\vv{s_i}\ansubst{\sigma_\Delta}{}\real A$
    \end{itemize}
    
    From this we can conclude that:
    
    \begin{itemize}
        \item $\sdom{\Gamma,\Delta}\subseteq \FV{\gencase{\vv{t}}{\vv{v_1}}{\vv{v_n}}{\vv{s_1}}{\vv{s_n}}}$
        \item $\FV{\gencase{\vv{t}}{\vv{v_1}}{\vv{v_n}}{\vv{s_1}}{\vv{s_n}}}\subseteq \dom{\Gamma,\Delta}$
    \end{itemize}
    
    Then, given $\sigma\in\sem{\Gamma,\Delta}$, we have that $\ansubst{\sigma}{}=\ansubst{\sigma_\Gamma}{}\ansubst{\sigma_\Delta}{}$ for some $\sigma_\Gamma\in\sem{\Gamma}$ and $\sigma_\Delta\in\sem{\Delta}$. Using the first hypothesis we have that, $\vv{t}\ansubst{\sigma_\Gamma}{}\real\sharp\genbasis{\vv{v_i}}{i=1}{n}$, then $\vv{t}\ansubst{\sigma_\Gamma}{}\eval e^{i\theta_1} \vv{u}\equiv e^{i\theta_1} (\sum_{i=1}^{n}\beta_i \vv{v_i})$ where $\sum_{i=1}^{n}|\beta_i|^2=1$. From the second hypothesis we have that $\vv{s_i}\ansubst{\sigma_\Delta}{}\eval e^{i\rho_i} \vv{u_i}\in\sem{A}$ for $i\in\{1,\dotsb ,n\}$ and $u_i\perp u_j$ if $i\neq j$. Therefore:

    \begin{align*}
        (&\gencase{\vv{t}}{\vv{v_1}}{\vv{v_n}}{\vv{s_1}}{\vv{s_n}})\ansubst{\sigma}{}\\ 
        &= (\gencase{\vv{t}}{\vv{v_1}}{\vv{v_n}}{\vv{s_1}}{\vv{s_n}})\ansubst{\sigma_\Gamma}{}\ansubst{\sigma_\Delta}{}\\
        &=(\sum_{i=1}^{n}\alpha_i \gencase{\vv{t}[\sigma_{\Gamma i}]}{\vv{v_1}}{\vv{v_n}}{\vv{s_1}}{\vv{s_n}})\ansubst{\sigma_\Delta}{} \\
        &\equiv (\gencase{\sum_{i=1}^{n} \alpha_i \vv{t}[\sigma_{\Gamma i}]}{\vv {v_1}}{\vv{v_n}}{\vv{s_1}}{\vv{s_n}})\ansubst{\sigma_\Delta}{}\\
        &=(\gencase{\vv{t}\ansubst{\sigma_\Gamma}{}}{\vv{v_1}}{\vv{v_n}}{\vv{s_1}}{\vv{s_n}})\ansubst{\sigma_\Delta}{}\\
        &\eval(\gencase{e^{i\theta_1} \vv{u}}{\vv{v}}{\vv{w}}{\vv{s_1}}{\vv{s_2}})\ansubst{\sigma_\Delta}{}\\
        &\lra e^{i\theta_1} (\sum_{i=1}^{n}\beta_i s_i)\ansubst{\sigma_\Delta}{}\\
        &= e^{i\theta_1} (\sum_{j=1}^{n}\delta_j (\sum_{i=1}^{n}\beta_i \vv{s_i})[\sigma_{\Delta j}])\\
        &\equiv e^{i\theta_1} (\sum_{i,j=1}^{n}\beta_i\delta_j\vv{s_i}[\sigma_{\Delta j}])\\
        &= e^{i\theta_1} (\sum_{i=1}^{n}\beta_i \vv{s_i}\ansubst{\sigma_\Delta}{})\\
        &\eval e^{i\theta_1} (\sum_{i=1}^{n}\beta_i e^{i\rho_i} \vv{u_i})
    \end{align*}
    
    It remains to be seen that: $\|\sum_{i=1}^{n}\beta_i e^{i\rho_i} \vv{u_i}\|=1$:
    \begin{align*}
        \|\sum_{i=1}^{n}\beta_i e^{i\rho_i} \vv{u_i}\| &= \scal{\sum_{i=1}^{n}\beta_i e^{i\rho_i} \vv{u_i}}{\sum_{i=1}^{n}\beta_i e^{i\rho_i} \vv{u_i}}\\
        &= \sum_{i,j=1}^{n}\overline{\beta_i e^{i\rho_i}}\beta_j e^{i\rho_j} \scal{\vv{u_i}}{\vv{u_j}}\\
        &= \sum_{i=1}^{n}\overline{\beta_i e^{i\rho_i}}\beta_i e^{i\rho_i} \scal{\vv{u_i}}{\vv{u_i}}\\
        &\qquad + \sum_{i,j=1; i\neq j}^{n}\overline{\beta_i e^{i\rho_i}}\beta_j e^{i\rho_j} \scal{\vv{u_i}}{\vv{u_j}}\\
        &= \sum_{i=1}^{n}|\beta_i|^2 |e^{i\rho_i}|^2  + 0\\
        &= \sum_{i=1}^{n}|\beta_i|^2 = 1
    \end{align*}

    Then we can conclude that $\sum_{i=1}^{n}\beta_i e^{i\rho_i}\vv{u_i}\in\sem{\sharp A}$ and finally: $(\gencase{\vv{t}}{\vv{v_1}}{\vv{v_n}}{\vv{s_1}}{\vv{s_n}})\ansubst{\sigma}{}\real\sharp A $
  
    \item[Sum] If the hypothesis is valid then for every $i$, $\sdom{\Gamma}\subseteq\FV{\vv{t_i}}\subseteq\dom{\Gamma}$.
    
    From this we can conclude that $\sdom{\Gamma}\subseteq\sum_{i=1}^{n}\alpha_i \vv{t_i}\subseteq\dom{\Gamma}$. Given $\sigma\in\sem{\Gamma}$, we have for every $i$, $\vv{t_i}\ansubst{\sigma}{}\eval e^{i\rho_i} \vv{v_i}$ where $\vv{v_i}\in\sem{A}$. Moreover, for every $i\neq j$, $\vv{v_i}\perp\vv{v_j}$ and $\sum_{i=1}^{n}|\alpha_i|^2=1$. Then:
    \begin{align*}
    (\sum_{i=1}^{n}\alpha_i\vv{t_i})\ansubst{\sigma}{} 
    &= \sum_{j=1}^{m}\beta_j(\sum_{i=1}^{n}\alpha_i \vv{t_i})[\sigma_j]\\
    &\equiv \sum_{i=1}^{n} \alpha_i \sum_{j=1}^{m} \beta_j \vv{t_i}[\sigma_j]\\
    &=\sum_{i=1}^{n} \alpha_i \vv{t_i}\ansubst{\sigma}{}\\
    &\eval \sum_{i=1}^{n} \alpha_i e^{i\rho_i} \vv{v_i}\\
    \end{align*}

    It remains to be seen that $\|\sum_{i=1}^{n} \alpha_i e^{i\rho_i} \vv{v_i}\|=1$:
    \begin{align*}
    &\|\sum_{i=1}^{n} \alpha_i e^{i\rho_i} \vv{v_i}\| \\
    &=\scal{\sum_{i=1}^{n} \alpha_i e^{i\rho_i}\vv{v_i}}{\sum_{i=1}^{n} \alpha_i e^{i\rho_i} \vv{v_i}}\\
    &= \sum_{i=i}^{n}\sum_{j=1}^{n} \overline{\alpha_i e^{i\rho_i}}\alpha_j e^{i\rho_j} \scal{\vv{v_i}}{\vv{v_j}}\\
    &=\sum_{i=1}^{n} \overline{\alpha_i e^{i\rho_i}}\alpha_i e^{i\rho_i} \scal{\vv{v_i}}{\vv{v_i}} + \sum_{\substack{i,j=1\\i\neq j}}^{n} \overline{\alpha_i e^{i\rho_i}}\alpha_j e^{i\rho_j} \scal{\vv{v_i}}{\vv{v_j}}\\
    &=\sum_{i=1}^{n}|\alpha_i|^2 |e^{i\rho_i}|^2+ 0\\
    &=\sum_{i=1}^{n}|\alpha_i|^2 = 1\\
    \end{align*}

    Then we can conclude that $\sum_{i=1}^{n}\alpha_i e^{i\rho_i}\vv{v_i}\in\sem{\sharp A}$ and finally $(\sum_{i=1}^{n}\alpha_i\vv{t_i})\ansubst{\sigma}{}\real\sharp A$.

    \item[Contr] If the hypothesis is valid, we have that $\sdom{\Gamma, x^X:\basis{X}, y^X:\basis{X}}\subseteq\FV{\vv{t}}\subseteq\dom{\Gamma,x^X:\basis{X}, y^X:\basis{X}}$ and given $\sigma\in\sem{\Gamma, x^X:\basis{X}, y^X:\basis{X}}$, then $\vv{t}\ansubst{\sigma}{}\in\sem{B}$. Then, we have that $\sdom{\Gamma, x^X:\basis{X}, y^X:\basis{X}}=\sdom{\Gamma, x^X:\basis{X}}$. Therefore:
    
    \[
    \sdom{\Gamma, x^X:\basis{X}}\subseteq\FV{\vv{t}}[x/y]\subseteq\dom{\Gamma, x^X:\basis{X}}
    \]

    Given $\sigma\in\sem{\Gamma, x^X:\basis{X}}$, we observe that $\ansubst{\sigma}{}=\ansubst{\vv{v}/x}{X}\ansubst{\sigma_\Gamma}{}$ with $\sigma_\Gamma\in\sem{\Gamma}$ and $\vv{v}\in\sem{\basis{X}}$. Since $\vv{v}\in\sem{\basis{X}}$, we know that $\vv{t}[\vv v/z] =\vv{t}\ansubst{\vv{v}/z}{X}$ for any variable $z$. Then we have:
    \begin{align*}
        \vv{t}[x/y]\ansubst{\sigma}{} &= \vv{t}[x/y]\ansubst{\vv{v}/x}{X}\ansubst{\sigma_\Gamma}{}\\
        &=\vv{t}[x/y][\vv{v}/x]\ansubst{\sigma_\Gamma}{}\\
        &=\vv{t}[\vv{v}/y][\vv{v}/x]\ansubst{\sigma_\Gamma}{}\\
        &=\vv{t}\ansubst{\vv{v}/y}{X}\ansubst{\vv{v}/x}{X}\ansubst{\sigma_\Gamma}{}    
    \end{align*}
    
    Since $\ansubst{\vv{v}/y}{X}\ansubst{\vv{v}/x}{X}\ansubst{\sigma}{}\in\sem{\Gamma, x^X:\basis{X}, y^X:\basis{X}}$, we get: $\vv{t}\ansubst{\vv{v}/y}{X}\ansubst{\vv{v}/x}{X}\\\ansubst{\sigma_\Gamma}{}\eval e^{i\theta}\vv{w}\in\sem{B}$
    Then we can finally conclude that $\vv{t}[x/y]\ansubst{\sigma}{}\real B$.

    \item[Weak] Given $\sigma\in\sem{\Gamma,x^X:\basis{X}}$, we observe that $\ansubst{\sigma}=\ansubst{\sigma_\Gamma}{}\ansubst{\vv{v}/x}{X}$ for some $\sigma_\Gamma\in\sem{\Gamma}$ and $\vv{v}\in\sem{\basis{X}}$. Using the first hypothesis, we know that $\vv{t}\ansubst{\sigma_\Gamma}{}\eval e^{i\theta}\vv{w}$ where $\vv{w}\in\sem{B}$. Then we have:
    \[
    \vv{t}\ansubst{\sigma}{}=\vv{t}\ansubst{\sigma_\Gamma}{}\ansubst{\vv{v}/x}{X}\eval e^{i\theta}\vv{w}\ansubst{\vv{v}/x}{X}
    \]
    Since $\vv{v}\in\sem{\basis{X}}$, $\vv{w}\ansubst{\vv{v}/x}{X}=\vv{w}[\vv{v}/x]=\vv{w}$ and $\vv{w}\in\sem{B}$, then we can finally conclude that $\vv{t}\ansubst{\sigma}{}\real B$.
    
    \item[Sub] Trivial since the set of realizers of ${A}$ is included in the set of realizers of ${B}$. 

    \item[Equiv] It follows from definition and the fact that the reduction commutes with the congruence relation.
    
    \item[Phase] It follows from the definition of type realizers.
    \end{description}
\end{proof}

\begin{lemma}[Substitution]\label{lem:Substitution}
  Let $\Gamma$, $\Delta$ be contexts, $A$ and $B$ types, and $X$ an orthonormal basis.
  If
  \(
    \TYP{\Gamma,\,x^{X}\!:\!A}{\vv{t}}{B}
    \quad\text{and}\quad
    \TYP{\Delta}{\vv{v}}{A}
  \) can be derived using the set of rules in \Cref{tab:TypingRules},
  and the substitution $\vv{t}\ansubst{\vv{v}/x}{X}$ is defined,
  then
  \(
    \TYP{\Gamma,\,\Delta}{\vv{t}\ansubst{\vv{v}/x}{X}}{B}
  \) can also be derived by the same set of rules.
\end{lemma}
\begin{proof}
  By induction on $\vv{t}$.
  \begin{description}
    \item[$x\not\in\FV{\vv{t}}$:] Then $x\not\in\sdom{\Gamma,\,x^{X}\!:\!A}$. By a straightforward generation lemma we have that $\TYP{\Gamma}{\vv{t}}{B}$, and we can derive $\TYP{\Gamma,\Delta}{\vv{t}}{B}$ via the $\textsc{Weak}$ rule. Notice that if $x\not\in\FV{\vv{t}}$, every variable in $\Delta$ has a non-linear type and $\vv{t}\ansubst{\vv{v}/x}{X}=\vv{t}$.
    
    \item[$\vv{t}=x$:] Then every variable in $\Gamma$ has a non-linear type and $A\leq B$. This means we can derive $\TYP{\Gamma,\Delta}{\vv{v}}{B}$, by rules $\textsc{Weak}$ and $\textsc{Sub}$.
    
    \item[$\vv{t}=(\Lam{y}{Y}{\vv{s}})$:] Then $\TYP{\Gamma,\,y^{Y}\!:\!C,\,x^{X}\!:\!A}{\vv{s}}{D}$, with $C\Arr D\leq B$. By induction hypothesis, $\TYP{\Gamma,\,y^{Y}\!:\!C}{\vv{s}\ansubst{\vv{v}/x}{X}}{D}$. This means we can derive $\TYP{\Gamma,\Delta}{(\Lam{y}{Y}{\vv{s}})\ansubst{\vv{v}/X}{X}}{B}$, by rules $\textsc{Sub}$ and $\textsc{UnitLam}$.
    
    \item[$\vv{t}=\vv{s_1}\,\vv{s_2}$:] Then we have that $\TYP{\Gamma_1}{\vv{s_1}}{C\Arr D}$, and $\TYP{\Gamma_2}{\vv{s_2}}{C}$ with $D\leq B$ and $\Gamma_1,\Gamma_2 = \Gamma,x^{X}\!:\!A$. We consider the case where $\Gamma_1 = \Gamma_1',\,x^{X}\!:\!A$. Then, by inductive hypothesis $\TYP{\Gamma_1'\Delta}{\vv{s_1}\ansubst{\vv{v}/x}{X}}{C\Arr D}$. This means we can derive $\TYP{\Gamma_1',\Gamma_2,\Delta}{(\vv{s_1}\,\vv{s_2})\ansubst{\vv{v}/x}{X}}{B}$, by rules $\textsc{Sub}$ and $\textsc{App}$. The case where $\Gamma_2 = \Gamma_2',\,x^{X}\!:\!A$ is analogous.
    
    \item[$\vv{t}=\Pair{s_1}{s_2}$:] Then we have that $\TYP{\Gamma_1}{s_1}{C}$, and $\TYP{\Gamma_2}{s_2}{D}$ with $C\times D\leq B$ and $\Gamma_1,\,\Gamma_2 = \Gamma,x^{X}\!:\!A$. We consider the case where $\Gamma_1 = \Gamma_1',\,x^{X}\!:\!A$. Then, by inductive hypothesis $\TYP{\Gamma_1',\,\Delta}{s_1}{C}$. This means we can derive $\TYP{\Gamma_1',\,\Gamma_2,\,\Delta}{\Pair{s_1}{s_2}\ansubst{\vv{v}/x}{X}}{B}$, by rules $\textsc{Sub}$ and $\textsc{Pair}$. The case where $\Gamma_2 = \Gamma_2',\,x^{X}\!:\!A$ is analogous.
    
    \item[$\vv{t}=\LetP{y}{Y}{z}{Z}{\vv{s_1}}{\vv{s_2}}$:] Then we have two possibilities:
      \begin{enumerate}
        \item\label{case:subs_letp_1} Either, $\TYP{\Gamma_1}{\vv{s_1}}{C\times D}$, and $\TYP{\Gamma_2,\,y^{Y}\!:\!C,\,z^{Z}\!:\!D}{\vv{s_2}}{E}$ with $E\leq B$.
        \item\label{case:subs_letp_2} Or, $\TYP{\Gamma_1}{\vv{s_1}}{\sharp(C\times D)}$, and $\TYP{\Gamma_2,\,y^{Y}\!:\!\sharp{C},\,z^{Z}\!:\!\sharp{D}}{\vv{s_2}}{E}$ with $\sharp E\leq B$.
      \end{enumerate}
      In either way, we consider the case where $\Gamma_1 = \Gamma_1',\,x^{X}\!:\!A$. Then, by inductive hypothesis $\TYP{\Gamma_1'\,\Delta}{\vv{s_1}\ansubst{\vv{v}/x}{}}{C\times D}$ in case \ref{case:subs_letp_1} ($\sharp (C\times D)$ in case \ref{case:subs_letp_2}). This means we can derive $\TYP{\Gamma_1',\,\Gamma_2,\Delta}{(\LetP{y}{Y}{z}{Z}{\vv{s_1}}{\vv{s_2}})\ansubst{\vv{v}/x}{X}}{B}$ by rules $\textsc{Sub}$ and $\textsc{LetPair}$ (or, $\textsc{LetTens}$). The case where $\Gamma_2 = \Gamma_2',\,x^{X}\!:\!A$ is analogous.

    \item[$\vv{t}=\gencase{\vv{s}}{\vv{w_1}}{\vv{w_n}}{\vv{r_1}}{\vv{r_n}}$:] Then we have two possibilities:
    \begin{enumerate}
      \item\label{case:subs_case_1} Either, $\TYP{\Gamma_1}{\vv{s_1}}{\basis{\{\vv{v_i}\}_{i=1}^n}}$, and for all $i\in\{0,\dotsb,n\}$, $\TYP{\Gamma_2}{\vv{s_i}}{C}$ with $C\leq B$.
      \item\label{case:subs_case_2} Or, $\TYP{\Gamma_1}{\vv{s_1}}{\sharp\basis{\{\vv{v_i}\}_{i=1}^n}}$, and for all $i,j\in\{0,\dotsb,n\}$, with $i\neq j$ $\SORTH{\Gamma_2}{\vv{s_i}}{\vv{s_j}}{C}$ with $\sharp C\leq B$.
    \end{enumerate}
    In either way, we consider the case where $\Gamma_1 = \Gamma_1',\,x^{X}\!:\!A$. Then, by inductive hypothesis $\TYP{\Gamma_1'\,\Delta}{\vv{s_1}\ansubst{\vv{v}/x}{}}{\basis{\{\vv{v_i}\}_{i=1}^n}}$ in case \ref{case:subs_case_1} ($\sharp \basis{\{\vv{v_i}\}_{i=1}^n}$ in case \ref{case:subs_case_2}). This means we can derive $\TYP{\Gamma_1',\,\Gamma_2,\Delta}{(\gencase{\vv{s}}{\vv{w_1}}{\vv{w_n}}{\vv{r_1}}{\vv{r_n}})\ansubst{\vv{v}/x}{X}}{B}$ by rules $\textsc{Sub}$ and $\textsc{Case}$ (or, $\textsc{UnitCase}$). Since orthogonality is preserved by substitutions in $\sem{\Gamma_2}$, the case where $\Gamma_2 = \Gamma_2',\,x^{X}\!:\!A$ is analogous.

    \item[$\vv{t}=\sum_{i=1}^{n}\alpha_i \vv{s_i}$:] Then we have two posibilities:
    \begin{enumerate}
      \item\label{case:subs_sum_1} For all $i\in\{0,\dotsb, n\}, \vv{s_i}=(\Lam{y}{Y}{\vv{r_i}})$, and:\\
      $\TYP{\Gamma,\,x^{X}\!:\!A}{\sum_{i=1}^{n}\alpha_i(\Lam{y}{Y}{\vv{r_i}})}{C\Arr D}$ with $C\Arr D\leq B$.
      \item\label{case:subs_sum_2} For all $i\in\{0,\dotsb, n\}$, with $i\neq j$, $\SORTH{\Gamma,\,x^{X}\!:\!A}{\vv{s_i}}{\vv{s_j}}{C}$, $\sum_{i=1}^{n}|\alpha_i|^2=1$, and $\sharp C\leq B$.
    \end{enumerate}
    
    In case \ref{case:subs_sum_1}, by inductive hypothesis we have that:\\
    $\TYP{\Gamma,\,\Delta,\,y^{Y}\!:\!C}{\sum_{i=1}^{n}\alpha_i \vv{r_i}\ansubst{\vv{v}/x}{X}}{D}$. This means we can derive \\$\TYP{\Gamma,\Delta}{\sum_{i=1}^{n}\alpha_i (\Lam{y}{Y}{\vv{r_i}})\ansubst{\vv{v}/x}{X}}{B}$ by rules $\textsc{Sub}$ and $\textsc{UnitLam}$.

    In case \ref{case:subs_sum_2}, by inductive hypothesis we have that for all $i\in\{0,\dotsb,n\}$, and $i\neq j$ $\SORTH{\Gamma,\Delta}{\vv{s_i}\ansubst{\vv{v}/x}{X}}{\vv{s_j}\ansubst{\vv{v}/x}{X}}{C}$. Since orthogonality is preserved by substitutions in $\sem{\Gamma,\Delta}$, this means we can derive $\TYP{\Gamma,\Delta}{\sum_{i=1}^n\alpha_i \vv{s_i}\ansubst{\vv{v}/x}{X}}{B}$ by rules $\textsc{Sub}$ and $\textsc{Sum}$.

    \item[$\vv{t}=\alpha\vv{s}$:] Then $\alpha=e^{i\theta}$, and $\TYP{\Gamma,\,x^{X}\!:\!A}{\vv{s}}{C}$ with $C\leq B$. By induction hypothesis, $\TYP{\Gamma}{\vv{s}\ansubst{\vv{v}/x}{X}}{C}$. This means we can derive $\TYP{\Gamma,\Delta}{e^{i\theta}\vv{s}\ansubst{\vv{v}/X}{X}}{\\B}$, by rules $\textsc{Sub}$ and $\textsc{Phase}$.
      \qedhere
  \end{description}
\end{proof}

\begin{restatetheorem}[Restatement of \Cref{thm:SubjectReduction}]
  If $\TYP{\Gamma}{\vv{t}}{A}$ can be derived using the set of rules in
  \Cref{tab:TypingRules} and $\vv{t}\to\vv{u}$, then
  $\TYP{\Gamma}{\vv{u}}{A}$ can also be derived by the same set of rules.
\end{restatetheorem}
\begin{proof}
  We proceed by induction on the derivation of the elementary reduction
  $\lraneq$. The congruence closure to obtain $\lra$ is handled trivially via
  the \textsc{Equiv} rule. We only give the basis cases as the inductive cases (the contextual cases) are straightforward.
  \begin{itemize}
    \item Let $(\Lam{x}{X}{\vv t})\vv v \lraneq \vv t\ansubst{\vv v/x}{X}$:
      Assume $\TYP{\Gamma,\Delta}{(\Lam{x}{X}{\vv t})\,\vv v}{B}$. By
      \textsc{App}, we have $\TYP{\Gamma}{\Lam{x}{X}{\vv t}}{A\Arr B}$ and
      $\TYP{\Delta}{\vv v}{A}$.  From $\TYP{\Gamma}{\Lam{x}{X}{\vv t}}{A\Arr B}$
      and \textsc{UnitLam}, we get $\TYP{\Gamma,x^{X}:A}{\vv t}{B}$. By
      \Cref{lem:Substitution},
      using $\TYP{\Delta}{\vv v}{A}$ and the fact that $\vv t\ansubst{\vv
      v/x}{X}$ is defined, we conclude $\TYP{\Gamma,\Delta}{\vv t\ansubst{\vv
      v/x}{X}}{B}$, as required.

    \item Let $\LetP{x}{X}{y}{Y}{\vv v}{\vv t} \lraneq \vv t\ansubst{\vv
      v/x\otimes y}{X\otimes Y}$: Assume
      $\TYP{\Gamma,\Delta}{\LetP{x}{X}{y}{Y}{\vv v}{\vv t}}{C}$.  By
      \textsc{LetPair} we have $\TYP{\Gamma}{\vv v}{A\times B}$ and
      $\TYP{\Delta,\,x^{X}:A,\,y^{Y}:B}{\vv t}{C}$.  Decompose $\vv v$ (in
      $X\otimes Y$) as required by the definition of $\ansubst{\vv v/x\otimes
      y}{X\otimes Y}$; by \Cref{lem:Substitution}
      applied twice
      (first for $x$, then for $y$), we conclude $\TYP{\Gamma,\Delta}{\vv
      t\ansubst{\vv v/x\otimes y}{X\otimes Y}}{C}$.

    \item Let $\gencase{\vv{v_k}}{\vv{v_1}}{\vv{v_n}}{\vv{t_1}}{\vv{t_n}}
      \lraneq \vv{t_k}$: Assume
      $\TYP{\Gamma,\Delta}{\gencase{\vv{v_k}}{\vv{v_1}}{\vv{v_n}}{\vv{t_1}}{\vv{t_n}}}{A}$.
      By \textsc{Case} we have $\TYP{\Gamma}{{\vv{v_k}}}{\basis{\{\vv
      v_i\}_{i=1}^n}}$ and $\TYP{\Delta}{\vv t_i}{A}$ for all $i$.  Thus, we are
      done.

    \item Let
      $\gencase{\sum_{i=1}^n\alpha_i\vv{v_i}}{\vv{v_1}}{\vv{v_n}}{\vv{t_1}}{\vv{t_n}}
      \lraneq \sum_{i=1}^n\alpha_i\vv{t_i}$: Assume
      $\TYP{\Gamma,\Delta}{\gencase{\sum_{i=1}^n\alpha_i\vv{v_i}}{\vv{v_1}}{\vv{v_n}}{\vv{t_1}}{\vv{t_n}}}{A}$.
      By \textsc{UnitCase} we have $\TYP{\Gamma}{\vv t}{\sharp\basis{\{\vv
      v_i\}}}$ and, for all $i\neq j$, $\SORTH{\Delta}{\vv t_i}{\vv t_j}{A}$.
      Then the reduct $\sum_{i=1}^{n}\alpha_i \vv{t_i}$ is typed by \textsc{Sum}
      as $\sharp A$ (using the orthogonality premises and the normalisation
      condition ensured by the semantics of $\sharp$), hence
      $\TYP{\Gamma,\Delta}{\sum_{i=1}^{n}\alpha_i \vv{t_i}}{\sharp A}$.
      \qedhere
  \end{itemize}
\end{proof}

\end{document}